\documentclass[12pt]{article}
\usepackage[title]{appendix}
\usepackage{rotating, lscape}
\usepackage{float}
\usepackage{graphicx, color}
\usepackage{threeparttable, booktabs, longtable, threeparttablex,
	tabularx, multirow, pdflscape}

\usepackage[utf8]{inputenc}
\usepackage[T1]{fontenc}
\usepackage{lmodern}
\usepackage{microtype}
\usepackage{setspace}

\usepackage{amssymb,amsmath,mathrsfs,amsthm,bbm}
\usepackage{accents}

\usepackage{afterpage}
\usepackage[font=small,format=plain,labelfont=bf,textfont=it]{caption}
\usepackage{subcaption}
\usepackage{ragged2e}
\usepackage{chngcntr}
\usepackage{enumerate}
\usepackage{url}
\AtBeginDocument{}

\usepackage{tkz-graph}
\usetikzlibrary{shapes.geometric,shapes,decorations,arrows,calc,
	arrows.meta,fit,positioning}
\tikzset{
	-Latex,auto,node distance =1 cm and 1 cm,semithick,
	state/.style ={ellipse, draw, minimum width = 0.7 cm},
	point/.style = {circle, draw, inner sep=0.04cm,fill,node contents={}},
	bidirected/.style={Latex-Latex,dashed},
	el/.style = {inner sep=2pt, align=left, sloped}
}

\usepackage[natbibapa]{apacite}

\usepackage{xcolor}
\usepackage{hyperref}
\hypersetup{
	colorlinks = true,
	urlcolor   = blue,
	linkcolor  = blue,
	citecolor  = blue
}

\usepackage[hmargin=1in,vmargin=1in]{geometry}
\allowdisplaybreaks

\theoremstyle{plain}
\newtheorem{proposition}{Proposition}

\newtheorem{example}{Example}

\theoremstyle{definition}
\newtheorem{assumption}{Assumption}

\title{\vspace{-1.5cm}When Do Treatment Changes Identify Causal Effects?}
\author{Martin Huber\\
	\small {University of Fribourg, Department of Economics} \smallskip }

\begin{document}
	
	\maketitle
	\thispagestyle{empty}
	\begin{abstract}
		\singlespacing
		This paper clarifies the identifying assumptions of causal inference based on treatment changes rather than levels, and their relation to conventional identification strategies. We characterize two structural models, with non-nested assumptions, under which treatment changes can be used for identification given observed covariates by differencing out time-constant confounders additive in the treatment equation. These assumptions are generally not nested with those of methods relying on treatment levels, such as selection on observables controlling for past outcomes, treatments, and covariates, or difference-in-differences, which differences outcomes rather than treatments over time. We show, however, that under a random-walk restriction on the treatment process, exploiting treatment changes for identification is equivalent to using treatment levels given lagged treatment. Such equivalence results motivate overidentification tests comparing methods based on treatment levels and changes. Under an alternative model without a random walk that instead rules out dynamic treatment effects (among other conditions), treatment changes can still serve as an instrument for a treatment effect constant given covariates. However, without a random walk, different identification strategies are generally not nested. In partially linear models with a treatment effect constant across periods, the non-nesting results imply that two-way fixed effects regression of the outcome change on the treatment change is doubly robust. This regression remains consistent if the time-constant confounder cancels under differencing in either the treatment or the outcome equation, given that the time-varying outcome unobservables are independent of those driving the treatment change. We characterize the causal models consistent with each method, run simulations, and present an empirical application to cigarette demand.
	\end{abstract}
	
	{\scriptsize
		\noindent \textbf{Address for correspondence:} Martin Huber, University of Fribourg, Bd.\ de P\'{e}rolles 90, 1700 Fribourg, Switzerland; martin.huber@unifr.ch.\\[0.3ex]
		\noindent \textbf{Acknowledgements:} The author is indebted to Michael Knaus and Henri Pfleiderer for valuable comments that helped to correct and substantially revise an earlier version of the draft, as well as to Cl\'{e}ment de Chaisemartin. The author used the AI assistant Claude (developed by Anthropic) as an aid for scientific writing and coding. Specifically, it was used to rephrase and improve author-written text and to optimize code.}
	
	\pagebreak
	\setcounter{page}{1}
	\normalsize
	\onehalfspacing
	\section{Introduction}\label{intro}
	
	Identifying causal effects in observational data remains a central challenge in empirical research. In the absence of random assignment, applied work frequently relies either on selection-on-observables assumptions, which require treatment assignment (such as a medical treatment dose or participation in a training program) to be independent of potential outcomes conditional on observed covariates, or on difference-in-differences (DiD) designs, which exploit repeated observations over time and assume parallel trends in potential outcomes, possibly conditional on covariates. Both approaches impose identifying assumptions with respect to the treatment level that are non-nested, as discussed in \citet{WebervanderLaanPetersen2015}, \citet{ChabeFerret2017}, \citet{xu2023causal}, and \citet{huber2024joint}.
	
	There exists a third practice that is not uncommon in applied work but rarely made theoretically explicit: assessing the effect of \emph{changes} in a dynamically evolving treatment rather than of specific treatment levels. Examples include studies of tax reforms that change marginal rates over time, labor market regulations that adjust minimum wages, education policies that shift compulsory schooling ages, and medical dosage protocols that modify treatment intensity. In such settings, researchers often implicitly assume that treatment changes rather than levels are quasi-randomly assigned conditional on observed covariates, without articulating precisely what structural conditions this requires or how it relates to conventional identification strategies. A common example is recoding the treatment as one in the case of a treatment increase and as zero if the treatment remains constant over time. \citet{CardKrueger1994}, for instance, apparently a classical DiD study, define the treatment group as observations experiencing an increase in the minimum wage (in the US state of New Jersey) and the control group as those where the minimum wage remains constant but is nonzero (in Pennsylvania).
	
	This paper provides a systematic analysis of identification based on treatment changes, aiming at differencing out time-constant confounders jointly affecting the outcome and the treatment from the treatment equation. We ask under what conditions on the underlying treatment and outcome processes the change in a unit's treatment from one period to the next can be treated as if it were randomly assigned, conditional on observed covariates, and what this requires beyond simply removing time-constant confounders. The treatment level of a unit in a specific time period consists of the treatment change added to the unit's treatment level in the previous period, and this previous level is itself a consequence of the unit's history, possibly related to the same unobserved factors that determine the outcome. Consequently, even after differencing out time-constant unobserved confounders from the treatment equation, the treatment change may still be related to the outcome indirectly, through the unit's prior treatment level. A related concern is raised by \citet*{arganaraz2025randomly} for first-difference
	regressions with a treatment effect that varies over time, whereas in our setting the dependence on the baseline treatment arises even under a constant effect.
	
	We characterize two structural models that address this. In Model A, the part of each period's treatment that is not explained by observed covariates and a time-constant confounder follows a random walk: each period's treatment shock is a fresh draw, unrelated to the entire history of past shocks and to the time-constant confounder. The treatment change therefore is unrelated to the unit's prior treatment level conditional on covariates, which permits the identification of heterogeneous treatment effects without strong restrictions on the outcome process, allowing e.g.\ for dynamic effects of earlier treatment levels on the outcome. Model B retains the same additive role for the time-constant confounder, so that it is removed by differencing the treatment just as in Model A, but drops the random-walk restriction on the time-varying part of the process, leaving its evolution otherwise unrestricted, and instead rules out dynamic treatment effects in the outcome equation. We show that this combination does not deliver the conditional independence assumption needed to identify treatment-change effects in general, since the treatment change can still be correlated with the unit's prior treatment level. Models A and B impose genuinely different and non-nested restrictions, and we show that they generally do not identify the same object.
	
	Although Model B does not deliver the treatment-change conditional independence assumption, it is not without content. In an outcome model that is linear in the treatment level with a coefficient that is homogeneous given covariates, Model B's restrictions suffice for the treatment change to be uncorrelated with the part of the outcome left unexplained by the treatment level. This permits applying an instrumental-variables strategy, see e.g.\ \citet{Wright1928} and \citet{Imbens+94}, with the treatment change serving as the instrument and the treatment effect recovered as a ratio of two coefficients of a reduced-form and first-stage regression. Once the treatment effect is heterogeneous across treatment levels, the same ratio continues to identify a weighted average of the marginal response, in the spirit of analogous results in the instrumental-variables literature on treatment effect heterogeneity, see \citet{AngristImbens95}. The weights are nonnegative if, conditional on observed covariates, the distribution of the treatment level net of the time-constant confounder is monotonic in the treatment change. However, the ratio generally does not identify the same conditional average marginal effect targeted directly by selection-on-observables or DiD under more general models permitting effect heterogeneity.
	
	We compare the conditions delivered by Models A and B to those underlying two conventional identification strategies that operate on treatment levels: selection-on-observables conditional on the lagged treatment, lagged outcome, and covariates, and DiD conditional on the lagged treatment and covariates but not the lagged outcome. Under Model A, all three identification strategies coincide and are informationally equivalent, since they all exploit the same underlying treatment shock. This equivalence motivates the construction of \citet{Hausman1978}-type overidentification tests by jointly considering methods based on treatment levels and treatment changes. Absent the structural assumptions of Model A, however, the linear regression coefficients these strategies rely on in practice need not coincide, as the corresponding identifying assumptions are generally non-nested with one another.
	
	Beyond these overidentification tests, the non-nesting results carry a double-robustness (DR) implication, see \citet{Robins+94}, for an estimator that relies on more than one identification strategy at once. Consider DiD and the treatment-change strategy, which aim to difference out the time-constant confounder from the outcome and the treatment equations, respectively. We show that, under specific structural conditions, the regression coefficient of the differenced outcome on the differenced treatment, as computed by the classical two-way fixed effects (TWFE) estimator, recovers the treatment effect if the time-constant confounder enters the treatment equation, or the outcome equation, in a way that cancels exactly under differencing, without requiring this for both equations simultaneously. This double robustness does not extend to models with treatment effects that are heterogeneous in the treatment level or in time. This structural notion of double robustness, which operates by combining two restrictions on the data-generating process, is related to the design-based and model-based notions of double robustness in \citet{arkhangelsky2022doubly} and \citet{arkhangelsky2021double}.
	
	We investigate the finite-sample behavior of the different identification strategies in a simulation study calibrated to a panel setting, confirming the non-nesting results, the rescaling requirement for the instrumental-variables estimator, and the double robustness of TWFE. As an empirical illustration, we estimate the price elasticity of cigarette consumption using the state-level panel of \citet{baltagi1994} covering 46 US states over the period 1963--1992. Applying the different identification strategies and the associated Hausman overidentification tests, we find a substantial gap between the estimate based on treatment changes and the estimates based on treatment levels and TWFE.
	
	The remainder of the paper is organized as follows. Section~\ref{setup} introduces the causal framework and the estimand of interest. Section~\ref{strucmod} presents Model A and Model B, the conditions each requires, and the relationship between them. Section~\ref{sec:comparison} compares the resulting conditions to those underlying DiD and selection-on-observables strategies. Section~\ref{sec:overid} discusses overidentification tests and the double-robustness result for TWFE. Section~\ref{sim} presents a simulation study and Section~\ref{app} an empirical application to the price elasticity of cigarette demand. Section~\ref{conc} concludes.
	\section{Causal Framework and Estimands}\label{setup}
	
	Most identification strategies for panel data or repeated cross-sections rely on differencing outcomes over time to remove unobserved time-constant confounders, also known as fixed effects; see e.g.\ the extensive difference-in-differences (DiD) literature \citep{Snow1855, Ashenfelter78}. This paper considers an alternative identification strategy based on treatment changes over time. Let $Y_T$ and $D_T$ denote the outcome and the treatment in period $T$, respectively. The treatment $D_T$ is presumably assigned at the beginning of period $T$, while the outcome $Y_T$ is measured at the end of that period. Let $X_T$ denote a vector of covariates observed in period $T$, also measured later in that period after treatment assignment $D_T$, which may contain both time-invariant and time-varying variables. For any variable $A$, let $\bar{A}_{T}$ denote its history up to period $T$, i.e.\ the collection of $A_t$ over all observed periods $t\le T$. In the examples and figures below, treatments and outcomes start in period $0$ and covariates in period $-1$, so that the first treatment $D_0$ can depend on the pre-treatment covariate $X_{-1}$. Throughout, capital letters denote random variables, while lowercase letters denote their realizations.
	
	To define our causal effect of interest (or estimand), we employ the potential outcomes framework \citep{Neyman23, Rubin74}. Standard potential outcomes are defined with respect to treatment levels: $Y_T(d_T)$ denotes the potential outcome at time $T$ under treatment level $d_T$ in that period. In contrast, we are interested in the causal effect of a change in treatment. For a treatment change value $\nabla d$, we consider the potential outcome $Y_T(D_{T-1}+\nabla d_T)$, where $\nabla d_T$ is a specific value of the treatment increment $\nabla D_T = D_T - D_{T-1}$. For notational convenience, we henceforth drop the time index and write the treatment increment as $\nabla d$. We focus on the average effect of a treatment change of size $\nabla d$ relative to no change in period $T=t$, denoted by $\Delta_t(\nabla d)$ and defined as
	\begin{align}\label{eff1}
		\Delta_t(\nabla d) = E\left[Y_t(D_{t-1}+\nabla d) - Y_t(D_{t-1})\right],
	\end{align}
	and more generally, on the average effect of two arbitrary treatment changes $\nabla d$ and $\nabla d'$ in period $T=t$, denoted by $\Delta_t(\nabla d, \nabla d')$ and defined as
	\begin{align}\label{eff2}
		\Delta_t(\nabla d, \nabla d')
		= E\left[Y_t(D_{t-1}+\nabla d) - Y_t(D_{t-1}+\nabla d')\right].
	\end{align}
	
	A crucial feature of this definition, and one that will be central to the analysis below, is that the potential outcome $Y_T(D_{T-1}+\nabla d)$ is evaluated at the unit's own, realized baseline treatment level $D_{T-1}$. This is what distinguishes the treatment-change estimand from a conventional treatment-level estimand $Y_T(d)$, which fixes $d$ independently of the unit. It is also worth noting that the same treatment increment may be realized under different baseline treatment levels across units, so treatment effects indexed by increments do not generally correspond to the effects of uniquely defined treatment levels. Such effects of treatment changes are nonetheless of direct policy relevance in settings where the treatment is dynamically adjusted, such as changes in tax rates, adjustments in public expenditures, or gradual modifications of medical dosages. Throughout we impose SUTVA \citep{Rubin80, Cox58} and rule out anticipation effects as is commonly assumed in DiD studies \citep{Lechner2010}.
	
	The central identifying assumption we study throughout the paper is the following conditional independence assumption, which requires the treatment change to be independent of the treatment-change potential outcome conditional on past covariates.
	\begin{assumption}\label{CIAdD}\text{(CIA-$\nabla D$)}\\
		$Y_T(D_{T-1}+\nabla d) \perp \nabla D_T \mid \bar{X}_{T-1}$ for all
		$\nabla d$ in the support of $\nabla D_T$.
	\end{assumption}
	Assumption~\ref{CIAdD} states that, after conditioning on the history of covariates observed prior to the treatment change, the treatment change $\nabla D_T$ is as good as randomly assigned. While the level of the treatment $D_T$ may be endogenous, identification hinges solely on the exogeneity of the treatment increment $\nabla D_T$ \emph{with respect to the unit's own treatment-change potential outcome}. We condition on $\bar{X}_{T-1}$ rather than contemporaneous covariates $X_T$ because the latter may include variables influenced by the treatment assignment in period $T$. Conditioning on post-treatment variables generally invalidates identification by introducing selection on outcomes, as e.g. discussed in \citet{angrist2009mostly}.
	
	As a further remark, note that if the lagged treatment $D_{T-1}$ is constant across units, then Assumption~\ref{CIAdD} coincides with the conventional selection-on-observables (or unconfoundedness) assumption for treatment evaluation formulated in terms of treatment levels, $Y_T(d) \perp D_T \mid \bar X_{T-1}$, see, for example, \citet{Im04}. A leading case is one in which no unit is treated in the earlier period, so that $D_{T-1}=0$ for all units and $\nabla D_T = D_T-0=D_T$. In this case, treatment changes and treatment levels coincide, and the baseline-dependence issue discussed in the introduction does not arise, precisely because there is no variation in the baseline to begin with. Therefore, formulating identifying assumptions in terms of treatment changes rather than levels, and confronting the baseline-dependence problem that this formulation introduces, is conceptually meaningful only when there is variation in $D_{T-1}$ across units.
	
	In addition to Assumption~\ref{CIAdD}, identification requires a common support condition, ensuring sufficient overlap in treatment changes across covariate histories.
	\begin{assumption}\label{comsup}\text{(Common Support)}\\
		$\Pr(\nabla D_T=\nabla d \mid \bar{X}_{T-1}) > 0$ for all $\nabla d$
		in the support of $\nabla D_T$.
	\end{assumption}
	Assumption~\ref{comsup} guarantees that each treatment change of interest occurs with positive probability for all relevant values of the conditioning variables, ensuring that the causal effects of treatment increments are empirically identifiable. While this condition is stated for discrete treatment changes, in the case of continuously distributed treatment changes it should be interpreted in terms of the conditional density: the density of $\nabla D_T$ given $\bar{X}_{T-1}$ must be bounded away from zero on its support, rather than requiring strictly positive probability mass at a point.
	
	Under Assumptions~\ref{CIAdD} and~\ref{comsup}, both the average treatment effect (ATE) of a treatment change and the conditional average treatment effect (CATE) given $\bar{X}_{T-1}$ are identified:
	\begin{align}
		\Delta_T(\nabla d, \nabla d') &= E\Big[ E[Y_T(D_{T-1}+\nabla d)\mid \bar{X}_{T-1}] - E[Y_T(D_{T-1}+\nabla d')\mid \bar{X}_{T-1}] \Big] \\ &= E\Big[ E[Y_T \mid \nabla D_T=\nabla d,\bar{X}_{T-1}] - E[Y_T \mid \nabla D_T=\nabla d',\bar{X}_{T-1}] \Big],\notag
	\end{align}
	where the first equality follows from the law of iterated expectations, and the second follows from Assumption~\ref{CIAdD}. Assumption~\ref{comsup} ensures that subpopulations with $\nabla d$ and $\nabla d'$ exist conditional on $\bar{X}_{T-1}$, so that the conditional expectations are well-defined. Analogous arguments can establish identification of distributional effects, such as quantile treatment effects, as considered for instance in \citet{Firpo03}.
	
	The remainder of the paper investigates under what structural conditions on the joint evolution of treatment, covariates, and unobservables Assumption~\ref{CIAdD} can be justified.
	\section{Structural Models for Treatment Changes}\label{strucmod}
	
	Assumption~\ref{CIAdD} is stated at a high level, and its plausibility is not transparent without reference to an underlying data-generating process. To understand the types of structural models in which this assumption is satisfied or violated, consider the following general equations for the outcome and treatment:
	\begin{align}\label{strucmodels}
		Y_T &= \mathcal{F}_Y \left(T, \bar{D}_{T}, \bar{X}_{T-1},
		\bar{V}_T, U\right), \notag\\
		D_T &= \mathcal{F}_D \left(T, \bar{D}_{T-1}, \bar{X}_{T-1},
		\bar{W}_T, U\right),
	\end{align}
	where $V_T$ and $W_T$ are time-varying unobservables and $U$ is a time-invariant unobservable or fixed effect. In general, Assumption~\ref{CIAdD} will not hold in model~\eqref{strucmodels} without further constraints. To see this, note that
	\begin{align}\label{genmod}
		\nabla D_T
		= \mathcal{F}_D(T,\bar{D}_{T-1},\bar{X}_{T-1},\bar{W}_T,U)
		- \mathcal{F}_D(T-1,\bar{D}_{T-2},\bar{X}_{T-2},\bar{W}_{T-1},U).
	\end{align}
	In nonlinear models, the fixed effect $U$ will typically not cancel, as its effect may interact with time. Consequently, $U$ remains a potential confounder that jointly affects $\nabla D_T$ and $Y_T$. Further violations of Assumption~\ref{CIAdD} can arise in several ways. First, if the time-varying unobservables $V_T$ and $W_T$ are statistically dependent, then $\nabla D_T$ and $Y_T$ may be statistically associated even after conditioning on past covariates. Second, if past treatments directly affect future outcomes --- i.e., there are dynamic treatment effects as e.g. considered in \citet{Ro86} and \citet{RoHeBr00} --- and also influence future treatments, then past treatments themselves become confounders of treatment changes and outcomes. In this case, conditioning on covariates alone may be insufficient for identification.
	
	A further point deserves emphasis at the outset, since it shapes everything that follows. The treatment change $\nabla D_T$ enters our analysis through the potential outcome $Y_T(D_{T-1}+\nabla d)$: the change of size $\nabla d$ is added on top of the treatment level $D_{T-1}$ that the unit actually had in the previous period. This previous treatment level $D_{T-1}$ is itself a random variable, determined by the unit's own history through equation~\eqref{strucmodels} at $T-1$, and it may depend on the same unobserved factors --- the fixed effect $U$ or the time-varying shocks $W_{T-1}$ --- that also affect the outcome. Removing $U$ from the treatment change $\nabla D_T$, as in~\eqref{genmod}, addresses only part of the problem: even after $U$ cancels, $\nabla D_T$ may still be informative about $D_{T-1}$, and through $D_{T-1}$ about the unobservables that drove the treatment process in earlier periods, which the potential outcome continues to depend on. Any structural restriction that delivers Assumption~\ref{CIAdD} must therefore ensure that the treatment change carries no information about the unit's prior treatment level, not only that it is free of the time-constant confounder.
	
	We now characterize two structural models. The first, Model A, restricts the time-varying part of the treatment process to follow a random walk. We show that this is sufficient to sever the link between the treatment change and the unit's prior treatment level entirely, delivering Assumption~\ref{CIAdD} for an unrestricted, dynamic, and heterogeneous treatment effect model. The second, Model B, retains only the restriction that the time-constant confounder enters the treatment equation additively, without any restriction on how the time-varying shock evolves over time, and additionally rules out dynamic treatment effects. We show that Model B does not deliver Assumption~\ref{CIAdD}: the link between the treatment change and the unit's prior treatment level remains open. Models A and B are not nested --- neither set of restrictions implies the other --- and we show that they do not, in general, identify the same object. Model B nonetheless delivers a related and useful condition: an instrumental-variables-type exogeneity condition that identifies a treatment effect that is constant given covariates, through the same logic as a standard instrumental-variables estimator.
	
	Model A restricts the treatment equation so that both the time-constant and time-varying unobservables are additively separable,
	\begin{align}\label{restrictedmodel}
		D_T = \mathcal{F}_D\left(T, \bar{X}_{T-1}\right) + U + W_T,
	\end{align}
	where $U$ is a scalar, time-constant unobservable or fixed effect, and $W_T$ is the time-varying unobservable specific to period $T$. We further assume that $W_T$ follows a random walk, implying that the treatment level in some period $T$ corresponds to the treatment level in the previous period $T-1$ plus a quasi-random shock that fully explains the change over time:
	\begin{align}\label{randomwalk}
		W_T = W_{T-1} + \varepsilon_T, \quad
		\varepsilon_T \perp \bar{W}_{T-1} \mid \bar{X}_{T-1},
	\end{align}
	where $\varepsilon_T$ is the period-$T$ shock to treatment. The independence in~\eqref{randomwalk} is required with respect to the \emph{entire} history $\bar{W}_{T-1}$, and not merely the immediately preceding shock $W_{T-1}$. This full-history property is what will let us show that the treatment change carries no information about the unit's prior treatment level, and it is correspondingly the feature that fails under Model B, introduced below.
	
	Under the treatment model restrictions in equations \eqref{restrictedmodel} and \eqref{randomwalk}, Assumption~\ref{CIAdD} holds if the time-specific shock $\varepsilon_T$ in the treatment equation is independent of the outcome-relevant unobservables, the time-invariant fixed effect, and the entire history of past treatment shocks, conditional on past covariates:
	\begin{align}\label{condind1}
		\{U, \bar V_T, \bar{W}_{T-1}\} \perp \varepsilon_T \mid \bar{X}_{T-1}.
	\end{align}
	To see this, note that under \eqref{restrictedmodel} the treatment change satisfies
	\begin{align}\label{treatsimp}
		\nabla D_T =& \mathcal{F}_D(T, \bar{X}_{T-1}) +W_{T}-\mathcal{F}_D(T-1, \bar{X}_{T-2})-W_{T-1},\notag\\ =& \mathcal{F}_D(T, \bar{X}_{T-1}) +\varepsilon_T-\mathcal{F}_D(T-1, \bar{X}_{T-2}),
	\end{align}
	where the first equality follows from cancellation of the fixed effect $U$, and the second from the random-walk specification \eqref{randomwalk}. Since $T$ is deterministic, the only stochastic element in the treatment change \eqref{treatsimp} conditional on $\bar{X}_{T-1}$ is the shock $\varepsilon_T$. Therefore, imposing \eqref{condind1} immediately gives
	\begin{align} \label{condind2}
		\{ \bar V_T, U, \bar{D}_{T-1} \} \perp \nabla D_T \mid \bar{X}_{T-1},
	\end{align}
	where the inclusion of $\bar{D}_{T-1}$ follows because $\bar{D}_{T-1}$ is a deterministic function of $(U,\bar{W}_{T-1},\bar{X}_{T-2})$, and \eqref{condind1} makes $\varepsilon_T$ independent of both $U$ and $\bar{W}_{T-1}$ --- the entire history that determines $\bar{D}_{T-1}$ --- given $\bar{X}_{T-1}$. This is the step that severs the link between $\nabla D_T$ and the unit's prior treatment level: under condition~\eqref{condind1}, $D_{T-1}$ itself is rendered independent of $\nabla D_T$, not merely the fixed effect $U$.
	
	Next, consider the potential outcome based on the unrestricted outcome model in \eqref{strucmodels} under a treatment change $\nabla d$:
	\begin{align}
		Y_T(\nabla d)
		= \mathcal{F}_Y\left(
		T, \bar D_{T-1}, D_{T-1} + \nabla d, \bar X_{T-1}, \bar V_T, U
		\right).
	\end{align}
	Conditional on $\bar X_{T-1}$, the stochastic elements in $Y_T(\nabla d)$ are $\bar D_{T-1}, \bar V_T$, and $U$ --- and, crucially, $\bar D_{T-1}$ enters explicitly, since the potential outcome is evaluated at the unit's own prior-level-plus-change $D_{T-1}+\nabla d$ rather than at a level fixed independently of the unit. By \eqref{condind2}, $\bar D_{T-1}, \bar V_T, U$ are jointly independent of $\nabla D_T$ conditional on $\bar X_{T-1}$. It follows that, for all $\nabla d$, $Y_T(\nabla d) \perp \nabla D_T \mid \bar X_{T-1}$, which is Assumption~\ref{CIAdD}. This result is formalized in Proposition~\ref{prop:modelA}, which permits an unrestricted, dynamic, and heterogeneous treatment effect model.
	
	\begin{proposition}\label{prop:modelA}(Identification under Model A).\\
		Under the treatment equation~\eqref{restrictedmodel}, the random walk with the full-history property~\eqref{randomwalk}, the unrestricted outcome equation in~\eqref{strucmodels}, and condition~\eqref{condind1}, Assumption~\ref{CIAdD} holds.
	\end{proposition}
	
	Figure~\ref{dag1} illustrates the causal structure underlying Model A using a directed acyclic graph (DAG); see, for instance, \citet{Pearl00} for background discussion. Solid nodes represent observed variables, dashed nodes represent unobserved variables, and arrows denote causal effects between variables. The left panel shows the model formulated in terms of treatment levels $D_0$ and $D_1$, whereas the right panel depicts the corresponding model for the treatment change $\nabla D_1 = D_1 - D_0$, with $D_1$ included explicitly as a node, fed by both $D_0$ and $\nabla D_1$, since $D_1 = D_0 + \nabla D_1$ is an identity. Both panels also include a direct arrow from $D_0$ to $Y_1$, reflecting that Model A's outcome equation permits a dynamic effect of the prior treatment level on the outcome beyond what is mediated by $D_1$.
	
	Conditional on $\bar{X}_0=\{X_{-1},X_0\}$, there are no variables that jointly affect $\nabla D_1$ and $Y_1$: in particular, neither the direct path from $D_0$ to $Y_1$ nor the path through $D_1$ confounds $\nabla D_1$ and $Y_1$, because under the full-history property~\eqref{randomwalk}, $\varepsilon_1$ --- the only stochastic driver of $\nabla D_1$ --- is independent of $D_0$ given $X_0$, so neither edge into $Y_1$ opens a back-door path into $\nabla D_1$. Model B, illustrated in Figure~\ref{dag2} below, instead rules out dynamic treatment effects altogether, so that the direct edge from $D_0$ to $Y_1$ is absent there and $D_0$ and $\nabla D_1$ affect $Y_1$ only through the mediator $D_1$. The analogous confounding nonetheless arises there because $\nabla D_1$ is correlated with $D_0$.
	
	\begin{figure}[!htp]
		\centering
		\caption{\label{dag1} Model A: DAG illustration of
			Proposition~\ref{prop:modelA}. Left: treatment levels. Right: treatment
			change. \bigskip}
		\begin{tikzpicture}[scale=0.8, transform shape]
			\node[state] (D0) at (0,0) {$D_0$};
			\node[state] (D)  [right =of D0] {$D_1$};
			\node[state] (Y1) [right =of D]  {$Y_1$};
			\node[state, dashed] (W0) [above =of D0] {$W_0$};
			\node[state, dashed] (W1) [above =of D]  {$W_1$};
			\node[state, dashed] (V1) [below =of D]  {$V_1$};
			\node[state, dashed] (U)  [below =of V1] {$U$};
			\node[state] (X0)  [below =of D0] {$X_0$};
			\node[state, dashed] (V0) [below =of X0] {$V_0$};
			\node[state] (Xm1) [left  =of X0] {$X_{-1}$};
			\path (D) edge (Y1); \path (U) edge (D0); \path (U) edge (Y1);
			\path (U) edge (X0); \path (U) edge[bend right=25] (D);
			\path (X0) edge (D); \path (D0) edge (X0);
			\path (D0) edge[bend left=30] (Y1);
			\path (X0) edge (Y1); \path (V1) edge (Y1); \path (X0) edge (V1);
			\path (W0) edge (D0); \path (W0) edge (W1); \path (W1) edge (D);
			\path (V0) edge (X0); \path (Xm1) edge (X0); \path (Xm1) edge (Y1);
			\path (Xm1) edge (D0); \path (Xm1) edge (V0); \path (Xm1) edge (D);
			\path (V0) edge (V1); \path (U) edge (Xm1);
		\end{tikzpicture}
		\quad
		\begin{tikzpicture}[scale=0.8, transform shape]
			\node[state] (D0) at (0,0) {$D_0$};
			\node[state] (D)  [right =of D0] {$\nabla D_1$};
			\node[state] (D1) [right =of D]  {$D_1$};
			\node[state] (Y1) [right =of D1] {$Y_1$};
			\node[state, dashed] (W0) [above =of D0] {$W_0$};
			\node[state, dashed] (W1) [above =of D]  {$\varepsilon_1$};
			\node[state, dashed] (V1) [below =of D1] {$V_1$};
			\node[state, dashed] (U)  [below =of V1] {$U$};
			\node[state] (X0)  [below =of D0] {$X_0$};
			\node[state, dashed] (V0) [below =of X0] {$V_0$};
			\node[state] (Xm1) [left  =of X0] {$X_{-1}$};
			\path (D1) edge (Y1); \path (U) edge (D0); \path (U) edge (Y1);
			\path (U) edge (X0); \path (X0) edge (D); \path (D0) edge (X0);
			\path (D0) edge[bend left=35] (Y1);
			\path (X0) edge (Y1);
			\path (V1) edge (Y1); \path (X0) edge (V1); \path (W0) edge (D0);
			\path (W1) edge (D); \path (V0) edge (X0);
			\path (D) edge (D1); \path (D0) edge[bend left=30] (D1);
			\path (Xm1) edge (X0); \path (Xm1) edge (Y1); \path (Xm1) edge (D0);
			\path (Xm1) edge (V0); \path (Xm1) edge (D); \path (V0) edge (V1);
			\path (U) edge (Xm1);
		\end{tikzpicture}
	\end{figure}
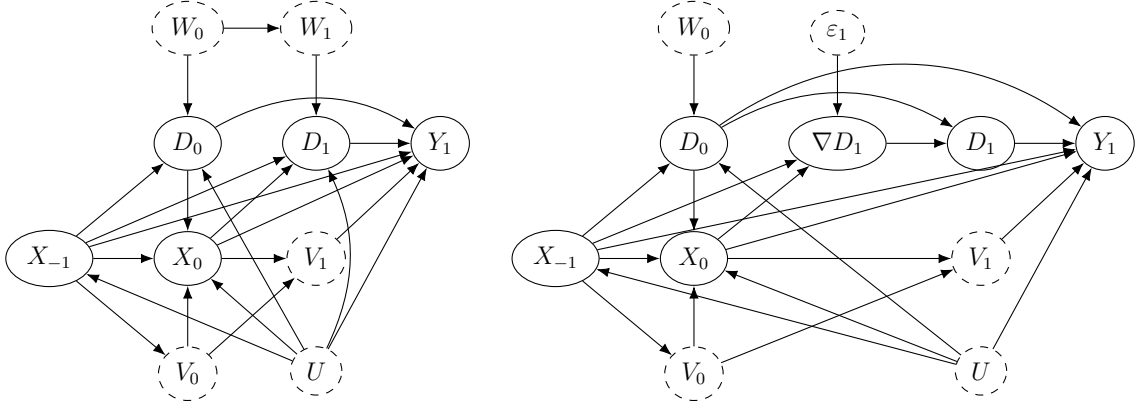
	It is natural to ask whether the random-walk restriction~\eqref{randomwalk} can be dispensed with, which is considered in Model B, introduced in the following discussion. Consider a weaker treatment equation that retains additive separability of the time-invariant unobservable, but otherwise leaves the evolution of the time-varying shock $W_T$ unrestricted:
	\begin{align}\label{restrictedmodel0}
		D_T &= \mathcal{F}_D\left(T,\bar{X}_{T-1}, \bar{W}_T \right) + U,
	\end{align}
	where $U$ is again a scalar, time-constant unobservable. Under~\eqref{restrictedmodel0}, the fixed effect still cancels in the treatment change:
	\begin{align}\label{treatsimp2}
		\nabla D_T =& \mathcal{F}_D\!\left(T,\bar{X}_{T-1}, \bar{W}_T \right)-\mathcal{F}_D\!\left(T-1,\bar{X}_{T-2}, \bar{W}_{T-1} \right),
	\end{align}
	so that, conditional on $\bar X_{T-1}$, $\nabla D_T$ depends stochastically only on $\bar W_T$. Suppose, in addition to~\eqref{restrictedmodel0}, that the outcome equation rules out dynamic treatment effects of earlier treatments on later outcomes, restricting the outcome to depend only on the contemporaneous treatment level and shock,
	\begin{align}\label{restrictedmodel2}
		Y_T = \mathcal{F}_Y\left(T, D_T, \bar{X}_{T-1}, V_T, U\right),
	\end{align}
	and impose the conditional independence condition
	\begin{align}\label{condind3}
		\{U, V_T\} \perp \bar W_{T} \mid \bar{X}_{T-1}.
	\end{align}
	Conditional on $\bar X_{T-1}$, equation~\eqref{treatsimp2} shows that $\nabla D_T$ depends stochastically only on $\bar W_T$, and condition~\eqref{condind3} makes $\{U,V_T\}$ independent of $\bar W_T$, and hence of $\nabla D_T$. This step is identical in form to the corresponding step in the proof of Proposition~\ref{prop:modelA}. The decisive difference arises one step later, when we examine the potential outcome itself.
	
	The potential outcome under~\eqref{restrictedmodel2} is $Y_T(D_{T-1}+\nabla d) = \mathcal{F}_Y(T,D_{T-1}+\nabla d,\bar X_{T-1},V_T,U)$. Conditional on $\bar X_{T-1}$, this depends on $D_{T-1}$ explicitly, and $D_{T-1} = \mathcal{F}_D(T-1,\bar X_{T-2},\bar W_{T-1}) + U$ depends on $\bar W_{T-1}$ --- which is part of the history that determines $\nabla D_T$ in~\eqref{treatsimp2}. Condition~\eqref{condind3} places no restriction on how $W_T$ evolves from its own past. Unless $W_T$ follows a random walk, $\nabla D_T$ in~\eqref{treatsimp2} therefore generally depends on $\bar W_{T-1}$. In the additive case $D_T=\mathcal{F}_D(T,\bar X_{T-1})+W_T+U$ with $W_T=\rho W_{T-1}+\varepsilon_T$, the treatment change contains the term $(\rho-1)W_{T-1}$, which vanishes only for $\rho=1$. Even under a random walk, the correlation between $W_{T-1}$ and $W_T$ is below one, so that the decisive quantity is the coefficient $\rho$. The case of independent shocks, $\rho=0$, is included. The link between the treatment change and the unit's prior treatment level is therefore not broken by~\eqref{restrictedmodel0}--\eqref{condind3}, such that Assumption~\ref{CIAdD} does not hold: ruling out dynamic treatment effects restricts the outcome equation, while the source of the problem lies in the treatment equation's evolution over time and is unaffected by restrictions on the outcome equation. The following example makes this failure explicit.
	
	\begin{example}[Model B does not imply Assumption~\ref{CIAdD}]\label{ex:modelBfails}
		Let $U, W_0, \varepsilon_1, V_1$ be mutually independent, with $\mathrm{Var}(W_0)>0$, and define
		\begin{align}
			\begin{gathered}
				D_0 = U+W_0, \qquad W_1=\rho W_0+\varepsilon_1, \qquad D_1 = U+W_1,\\
				Y_1(d) = \delta d + \gamma U + V_1, \qquad \delta\neq0,\ \rho\neq1.
			\end{gathered}
		\end{align}
		The time-varying treatment unobservable follows an autoregressive process of order one. This includes uncorrelated shocks ($\rho=0$) and excludes the random walk of Model A ($\rho=1$). This DGP satisfies every stated ingredient of Model B: $U$ is additively separable in the treatment equation and cancels from $\nabla D_1=D_1-D_0=(\rho-1)W_0+\varepsilon_1$; the outcome equation has no dynamic effect of $D_0$ once the contemporaneous level $D_1$ is fixed, as in~\eqref{restrictedmodel2}; and $\{U,V_1\}\perp(W_0,W_1)$, the no-covariate version of condition~\eqref{condind3}. For a fixed hypothetical change $a$,
		\begin{align}
			Y_1(D_0+a) = \delta(D_0+a)+\gamma U+V_1 = \delta a + \delta W_0 + (\delta+\gamma)U + V_1,
		\end{align}
		so that
		\begin{align}
			\mathrm{Cov}\big(\nabla D_1, Y_1(D_0+a)\big) &= \mathrm{Cov}\big((\rho-1)W_0+\varepsilon_1,\ \delta W_0+(\delta+\gamma)U+V_1\big) \notag\\
			&= \delta(\rho-1)\,\mathrm{Var}(W_0) \neq 0.
		\end{align}
		Assumption~\ref{CIAdD} therefore fails under Model B, for any $\rho\neq1$ and in particular for uncorrelated shocks with $\rho=0$. The treatment change and the potential outcome share $W_0$ through the unit's prior treatment level $D_0$: cancellation of $U$ from $\nabla D_1$ shows only that $\nabla D_1$ does not contain $U$, not that $\nabla D_1$ is independent of $D_0$, and indeed $\mathrm{Cov}(\nabla D_1,D_0) = (\rho-1)\mathrm{Var}(W_0) \neq 0$ here. Setting $\rho=1$ would make both covariances zero, which is the random walk of Model A.
	\end{example}
	
	A natural question is whether conditioning on the observed prior treatment level $D_{T-1}$ itself --- in addition to $\bar X_{T-1}$ --- repairs the conditional independence claim, i.e.\ whether $Y_T(D_{T-1}+\nabla d) \perp \nabla D_T \mid \bar X_{T-1}, D_{T-1}$ might hold under Model B even though Assumption~\ref{CIAdD} as stated does not. This is not a generic repair, because $D_{T-1}$ is itself jointly determined by $W_{T-1}$ and $U$ in the additive treatment model~\eqref{restrictedmodel0}: conditioning on $D_{T-1}$ induces an association between $W_{T-1}$ (which affects $\nabla D_T$) and $U$ (which affects the outcome), since knowing the sum $D_{T-1}=U+W_{T-1}$ makes a high realization of one component informative about a low realization of the other. Concretely, in the setting of Example~\ref{ex:modelBfails} with $U,W_0,\varepsilon_1$ Gaussian with variances $\sigma_U^2,\sigma_{W_0}^2,\sigma_{\varepsilon}^2$, one obtains
	\begin{align}
		\mathrm{Cov}(\nabla D_1, U \mid D_0) = (\rho-1)\,\mathrm{Cov}(W_0,U\mid D_0) = (1-\rho)\,\frac{\sigma_U^2\sigma_{W_0}^2}{\sigma_U^2+\sigma_{W_0}^2},
	\end{align}
	so that $\mathrm{Cov}(\nabla D_1,Y_1(D_0+a)\mid D_0) = \gamma(1-\rho)\,\sigma_U^2\sigma_{W_0}^2/(\sigma_U^2+\sigma_{W_0}^2)$, which is nonzero whenever $\gamma\neq0$ and $\rho\neq1$. Conditioning on $D_{T-1}$ therefore closes one confounding path while opening another, and is not a substitute for the random-walk restriction of Model A. In this context, it is worth noting that \citet*{arganaraz2025randomly} instead impose random assignment of the treatment change conditional on the prior treatment level
	$D_{T-1}$ directly (however, in a first-differenced outcome model) and control for its conditional mean. While this yields identification whenever that assumption holds, the additive structure of Model B does not provide a structural foundation for it, since $D_{T-1}$ is jointly determined by $W_{T-1}$ and $U$.
	
	Figure~\ref{dag2} illustrates Model B and Example~\ref{ex:modelBfails} using the same type of diagram as Figure~\ref{dag1}. The left panel depicts the causal model in terms of the original treatment levels $D_0$ and $D_1$, whereas the right panel represents the model when considering the treatment change $\nabla D_1 = D_1-D_0$, now including $D_1$ itself as an explicit node. Considering first the left graph, treatment $D_1$ and outcome $Y_1$ are confounded by the time-invariant unobservable $U$, even after conditioning on $X_0$. In the right graph, $U$ is, due to its additive separability, removed from the treatment-change node $\nabla D_1$ itself. Both $\nabla D_1$ and $D_0$ feed into $D_1$, reflecting the identity $D_1=D_0+\nabla D_1$, and $D_1$ in turn is the only node with a direct edge into $Y_1$, making explicit that neither $\nabla D_1$ nor $D_0$ affects the outcome other than through $D_1$.
	
	Furthermore, an edge $W_0 \to \nabla D_1$ is required in general, since, writing $\nabla D_1 = \mathcal{F}_D(1,\bar X_0,W_0,W_1) - \mathcal{F}_D(0,X_{-1},W_0)$, the dependence of this difference on $W_0$ cancels only in special cases. In Example~\ref{ex:modelBfails}, this dependence is $(\rho-1)W_0$, which cancels only under the random walk $\rho=1$ of Model A. The confounding path $\nabla D_1 \leftarrow W_0 \to D_0 \to D_1 \to Y_1$ exhibits the same problem as Example~\ref{ex:modelBfails}: $\nabla D_1$, which affects $Y_1$ through the treatment $D_1$, is correlated with $D_0$, the other determinant of $D_1$, through the shared dependence on $W_0$, which also affects $Y_1$ through $D_1$ and is thus a confounder of $\nabla D_1$.
	
	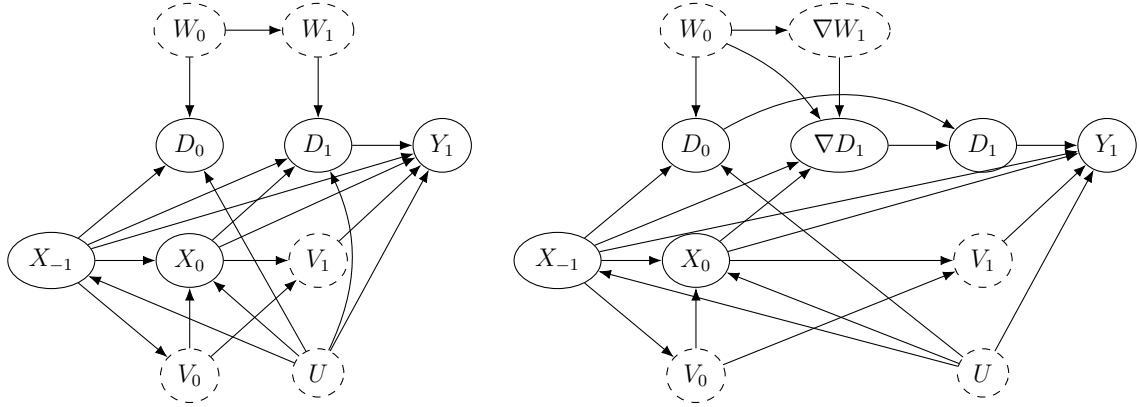
\begin{figure}[!htp]
		\centering
		\caption{\label{dag2} Model B: DAG illustration of
			Example~\ref{ex:modelBfails}. Left: treatment levels.
			Right: treatment change. \bigskip}
		\begin{tikzpicture}[scale=0.8, transform shape]
			\node[state] (D0) at (0,0) {$D_0$};
			\node[state] (D)  [right =of D0] {$D_1$};
			\node[state] (Y1) [right =of D]  {$Y_1$};
			\node[state, dashed] (W0) [above =of D0] {$W_0$};
			\node[state, dashed] (W1) [above =of D]  {$W_1$};
			\node[state, dashed] (V1) [below =of D]  {$V_1$};
			\node[state, dashed] (U)  [below =of V1] {$U$};
			\node[state] (X0)  [below =of D0] {$X_0$};
			\node[state, dashed] (V0) [below =of X0] {$V_0$};
			\node[state] (Xm1) [left  =of X0] {$X_{-1}$};
			\path (D) edge (Y1); \path (U) edge (D0); \path (U) edge (Y1);
			\path (U) edge (X0); \path (U) edge[bend right=25] (D); \path (V0) edge (V1);
			\path (X0) edge (D); \path (X0) edge (Y1); \path (V1) edge (Y1);
			\path (X0) edge (V1); \path (W0) edge (D0); \path (W0) edge (W1);
			\path (W1) edge (D); \path (V0) edge (X0);
			\path (Xm1) edge (X0); \path (Xm1) edge (Y1); \path (Xm1) edge (D0);
			\path (Xm1) edge (V0); \path (Xm1) edge (D); \path (U) edge (Xm1);
		\end{tikzpicture}
		\quad
		\begin{tikzpicture}[scale=0.8, transform shape]
			\node[state] (D0) at (0,0) {$D_0$};
			\node[state] (D)  [right =of D0] {$\nabla D_1$};
			\node[state] (D1) [right =of D]  {$D_1$};
			\node[state] (Y1) [right =of D1] {$Y_1$};
			\node[state, dashed] (W0) [above =of D0] {$W_0$};
			\node[state, dashed] (W1) [above =of D]  {$\nabla W_1$};
			\node[state, dashed] (V1) [below =of D1] {$V_1$};
			\node[state, dashed] (U)  [below =of V1] {$U$};
			\node[state] (X0)  [below =of D0] {$X_0$};
			\node[state, dashed] (V0) [below =of X0] {$V_0$};
			\node[state] (Xm1) [left  =of X0] {$X_{-1}$}; \path (V0) edge (V1);
			\path (D1) edge (Y1); \path (U) edge (D0); \path (U) edge (Y1);
			\path (U) edge (X0); \path (X0) edge (D); \path (X0) edge (Y1);
			\path (V1) edge (Y1); \path (X0) edge (V1); \path (W0) edge (W1);
			\path (W0) edge (D0); \path (W1) edge (D);
			\path (D) edge (D1); \path (D0) edge[bend left=30] (D1);
			\path (W0) edge[bend left=15] (D);
			\path (V0) edge (X0);
			\path (Xm1) edge (X0); \path (Xm1) edge (Y1); \path (Xm1) edge (D0);
			\path (Xm1) edge (V0); \path (Xm1) edge (D); \path (U) edge (Xm1);
		\end{tikzpicture}
	\end{figure}
	
	One notable property of Model B is the exclusion of a direct effect of $D_0$ on $X_0$, ensuring that $X_0$ is not affected by $W_0$ through $D_0$. Otherwise, $X_0$ would act as a collider on the path $U \rightarrow X_0 \leftarrow D_0 \leftarrow W_0$ according to the d-separation theorem for causal inference \citep{pearl1988probabilistic}. For this reason, conditioning on $X_0$ would induce a spurious association between $U$ and $W_0$, further compounding the violation of condition \eqref{condind3}. Importantly, this restriction implies that earlier treatments do not affect post-treatment covariates that are used as control variables for later treatment assignment, which may be restrictive in practice. One reason is that covariates may include earlier outcomes that are plausibly affected by past treatments.
	
	Although Model B does not deliver Assumption~\ref{CIAdD}, it is not without identifying content. We now show that, in an outcome model that is linear in the treatment level within each covariate cell, Model B delivers a condition that identifies the treatment effect through the same logic as an instrumental variable. Consider the partially linear model
	\begin{align}\label{partiallylinear}
		Y_T = \delta_T(\bar X_{T-1})\,D_T + \eta_T, \qquad \eta_T = f(U,\bar X_{T-1},V_T),
	\end{align}
	where $\delta_T(\bar X_{T-1})$ is the marginal effect of the treatment level, which may vary freely with covariates and across periods, and $\eta_T$ collects the time-invariant and time-varying unobservables additively, without any direct dependence on $D_{T-1}$ beyond the explicit term $\delta_T(\bar X_{T-1}) D_T$ --- this is the natural linear specialization of the no-dynamic-effects restriction~\eqref{restrictedmodel2}. Define the persistence factor
	\begin{align}\label{eq:lambda_def}
		\lambda_T(\bar x) = \frac{\mathrm{Cov}(\nabla D_T,D_T\mid \bar X_{T-1}=\bar x)}{\mathrm{Var}(\nabla D_T\mid \bar X_{T-1}=\bar x)},
	\end{align}
	which measures the degree to which the treatment level co-moves with its own change, conditional on covariates. We return to this object throughout the remainder of the paper.
	
	\begin{proposition}\label{prop:modelBiv}(Instrumental-variables exogeneity under Model B).\\
		Under the treatment equation~\eqref{restrictedmodel0} and condition~\eqref{condind3}, together with the outcome model~\eqref{partiallylinear}, the treatment change satisfies
		\begin{align}\label{eq:ivexog}
			\mathrm{Cov}(\nabla D_T, \eta_T \mid \bar{X}_{T-1}) = 0.
		\end{align}
		Consequently, at every $\bar x$ with $\lambda_T(\bar x) \neq 0$, the treatment effect $\delta_T(\bar x)$ is identified by the ratio
		\begin{align}\label{eq:ivratio}
			\delta_T(\bar x) = \frac{\mathrm{Cov}(\nabla D_T, Y_T \mid \bar{X}_{T-1}=\bar x)}{\mathrm{Cov}(\nabla D_T, D_T \mid \bar{X}_{T-1}=\bar x)} = \frac{\dfrac{\mathrm{Cov}(\nabla D_T, Y_T \mid \bar{X}_{T-1}=\bar x)}{\mathrm{Var}(\nabla D_T \mid \bar{X}_{T-1}=\bar x)}}{\dfrac{\mathrm{Cov}(\nabla D_T, D_T \mid \bar{X}_{T-1}=\bar x)}{\mathrm{Var}(\nabla D_T \mid \bar{X}_{T-1}=\bar x)}} = \frac{\delta_T(\bar x)\,\lambda_T(\bar x)}{\lambda_T(\bar x)}.
		\end{align}
	\end{proposition}
	
	\begin{proof}
		Fix $\bar X_{T-1}=\bar x$, so that $\delta_T(\bar x)$ is a constant. As shown in~\eqref{treatsimp2}, $\nabla D_T$ depends, conditional on $\bar x$, only on $\bar W_T$. The residual $\eta_T = f(U,\bar x,V_T)$ depends, conditional on $\bar x$, only on $(U,V_T)$. By condition~\eqref{condind3}, $\{U,V_T\} \perp \bar{W}_T \mid \bar{X}_{T-1}$, and since $\nabla D_T$ is a function of $\bar W_T$ alone given $\bar x$, this gives $\eta_T \perp \nabla D_T \mid \bar x$, which delivers the zero-covariance condition~\eqref{eq:ivexog}. This step does not involve $\delta_T(\bar x)$ at all. The ratio in~\eqref{eq:ivratio} then follows from~\eqref{partiallylinear} by writing $\mathrm{Cov}(\nabla D_T,Y_T\mid\bar x) = \delta_T(\bar x)\,\mathrm{Cov}(\nabla D_T,D_T\mid\bar x) + \mathrm{Cov}(\nabla D_T,\eta_T\mid\bar x)$ and substituting~\eqref{eq:ivexog}.
	\end{proof}
	
	Proposition~\ref{prop:modelBiv} and Proposition~\ref{prop:modelA} are not substitutes for one another: Appendix~\ref{app:decomposition} shows that, in the partially linear model~\eqref{partiallylinear}, requiring Assumption~\ref{CIAdD} to hold at the level of covariances forces a specific balance between the persistence factor $\lambda_T(\bar x)$ and the instrumental-variables condition~\eqref{eq:ivexog}, and that, except in a non-generic special case, this balance requires $\lambda_T(\bar x)=1$, which the random walk of Model A implies; Model B's restrictions deliver condition~\eqref{eq:ivexog} directly, without a random walk, but are on their own insufficient for Assumption~\ref{CIAdD} to hold, as shown in Example~\ref{ex:modelBfails}.
	
	The treatment change $\nabla D_T$ in Proposition~\ref{prop:modelBiv} has a natural instrumental-variables (IV) interpretation, see e.g.\ \citet{Wright1928} and \citet{Imbens+94}: it serves as an instrument for the treatment level $D_T$, with the persistence factor $\lambda_T(\bar x)$ playing the role of the first-stage coefficient conditional on $\bar X_{T-1}=\bar x$. Although $\nabla D_T$ is not a separately assigned instrument but an algebraic component of $D_T$ itself, since $D_T = D_{T-1} + \nabla D_T$ by definition, it satisfies the requirements of a valid instrument for $D_T$ in the relevant nonparametric sense: $D_{T-1}$, the remaining component of $D_T$, affects $Y_T$ only through $D_T$ rather than directly, so that it does not violate the exclusion restriction, and $\nabla D_T$ is correlated with $D_{T-1}$ without this constituting a violation of instrument validity, exactly as an instrument may be correlated with other determinants of the treatment as long as it is unrelated to the outcome except through the treatment.
	
	This connects the present setting to the literature on instrumental variables under heterogeneity in treatment effects, related to both observed and unobserved characteristics, as surveyed in \citet{mogstad2024instrumental}. Throughout, $\delta_T(\bar x)$ has been a single number within each covariate cell, governing a treatment effect that is linear in the treatment level. Empirical researchers are often interested in the opposite case, where the treatment effect is heterogeneous across treatment levels, with target parameter the conditional average marginal effect
	\begin{align}\label{eq:ame}
		\delta_t(\bar x) = E\left[\frac{\partial Y_t(d)}{\partial d}\bigg|_{d=D_t} \,\Big|\, \bar X_{t-1}=\bar x\right],
	\end{align}
	which reduces to the constant $\delta_T(\bar x)$ of the partially linear model~\eqref{partiallylinear} when the treatment effect is homogeneous. For a continuously varying treatment, \citet{AngristImbens95} show that the probability limit of two-stage least squares is, under such heterogeneity, a weighted average of the marginal treatment response along the support of the treatment, with weights determined by how the instrument shifts the distribution of the treatment, and not generally the conditional average marginal effect that selection-on-observables or DiD strategies target. The same logic applies to the ratio in equation~\eqref{eq:ivratio}: under treatment effect heterogeneity, it identifies a particular weighted average of the marginal effect, rather than the conditional average marginal effect $\delta_t(\bar x)$ that a selection-on-observables or DiD strategy targets directly. The weights add up to one and are nonnegative if the distribution of the treatment level net of the fixed effect, conditional on the covariates, is monotonic in the treatment change, i.e.\ shifts consistently upward or consistently downward as the treatment change increases. This condition can fail under Model B, since units with larger treatment changes may start from lower treatment levels when the treatment change and the prior treatment level are correlated.
	
	TWFE does not escape this either: its recovery of $\delta_t(\bar x)$ established in Section~\ref{sec:overid} is specific to a homogeneous treatment effect. If the treatment effect varies with the treatment level but not over time, the TWFE coefficient is a weighted average of the marginal effects over the treatment levels between $D_{T-1}$ and $D_T$, with weights that add up to one. The weights are nonnegative if the treatment change is independent of the prior treatment level conditional on the covariates, as under Model A, and can be negative otherwise. \citet*{callaway2024continuous} study TWFE under effect heterogeneity across the levels of a continuous treatment in a two-period design in which all units start from a zero dose. The treatment change is then trivially independent of the prior treatment level, in line with the nonnegative weights they obtain. Different identification strategies can therefore yield different point estimates not because any of their underlying assumptions is violated, but because each implicitly targets a different weighted average of an effect that varies with the treatment level: a form of treatment effect heterogeneity rather than a specification error.
	
	The IV structure of our framework connects naturally to the shift-share or Bartik instrument literature, see for instance \citet{bartik1991benefits} and \citet{borusyak2024practical}, in which the treatment change $\nabla D_t$ is constructed as the product of an aggregate shock (shift) and an initial treatment exposure (share). One example is \citet{autor2013china}, who assess the effect of rising Chinese import competition on regional employment in the US. They construct a shift-share instrument by multiplying industry-level import growth in non-US countries by regional industry employment shares reflecting each region's initial exposure to those shocks, and use this instrument to predict US regional import exposure (the treatment). In such settings, $\nabla D_t$ plays the role of the shift-share instrument, $D_t$ is the endogenous treatment, and $\hat\lambda$ is the first-stage coefficient. The exogeneity of $\nabla D_t$ under condition~\eqref{eq:ivexog} then corresponds to the conditions discussed in \citet{Borusyaketal2021} and \citet{GoldsmithPinkhametal2020}, who show that instrument validity requires either the shifts or the shares to be exogenous.
	
	Table~\ref{tab:comparison} summarizes the conditions associated with Models A and B and what each delivers.
	
	\begin{table}[ht]
		\centering
		\caption{Conditions and identification results associated with Models A and B
			\label{tab:comparison}}
		\small
		\begin{tabular}{lcc}
			\toprule
			& Model A & Model B \\
			\midrule
			Additive separability of $U$ in treatment equation
			& Required & Required \\
			Random-walk structure of treatment shocks & Required & Not required \\
			No dynamic treatment effects & Not required & Required (for Prop.~\ref{prop:modelBiv}) \\
			\midrule
			Delivers Assumption~\ref{CIAdD}, heterogeneous case & Yes (Prop.~\ref{prop:modelA}) & No (Example~\ref{ex:modelBfails}) \\
			Delivers condition~\eqref{eq:ivexog}, constant effect & Yes & Yes (Prop.~\ref{prop:modelBiv}) \\
			Recovers $\delta_t(\bar x)$ under heterogeneity & Yes & No, weighted average instead \\
			\bottomrule
		\end{tabular}
	\end{table}
	\section{Comparison with Treatment Level-Based Strategies}
	\label{sec:comparison}
	
	We now compare the conditions derived in the previous section to those underlying two conventional identification strategies that operate on treatment levels rather than changes: classical selection-on-observables, which treats treatment levels as exogenous conditional on observed variables when considering outcome levels, and DiD, which treats treatment levels as exogenous conditional on observed variables when considering outcome trends. Formally, selection on observables relies on the following conditional independence assumption (CIA), see e.g.\ \citet{Im04}:
	\begin{assumption}(CIA-$D$).\label{CIAD}\\
		$Y_T(d) \perp D_T \mid \bar{D}_{T-1},\, \bar{X}_{T-1},\, \bar{Y}_{T-1}$
		\quad for all $d$ in the support of $D_T$.
	\end{assumption}
	\noindent The parallel trends assumption underlying DiD can be formalized as follows:
	\begin{assumption}(CIA-$\nabla Y$).\label{CIAdY}\\
		$\nabla Y_T(d) \perp D_T \mid \bar{D}_{T-1},\, \bar{X}_{T-1}$
		\quad for all $d$ in the support of $D_T$.
	\end{assumption}
	
	Our approach based on treatment changes and the two approaches based on treatment levels all impose specific conditional independence assumptions, but differ in which source of variation is deemed exogenous for causal analysis and in the set of control variables. CIA-$\nabla D$ (Assumption~\ref{CIAdD}) treats the treatment change $\nabla D_T$ as exogenous conditional on $\bar{X}_{T-1}$. CIA-$D$ (Assumption~\ref{CIAD}) treats the treatment level $D_T$ as exogenous conditional on $\bar{D}_{T-1}$, $\bar{X}_{T-1}$, and $\bar{Y}_{T-1}$. CIA-$\nabla Y$ (Assumption~\ref{CIAdY}) imposes treatment exogeneity with respect to potential outcome changes, conditional on $\bar{D}_{T-1}$ and $\bar{X}_{T-1}$. It is worth emphasizing that for DiD identification based on CIA-$\nabla Y$, the conditioning set $\bar{X}_{T-1}$ must not include lagged outcomes $\bar{Y}_{T-1}$. If lagged outcomes are added to the conditioning set, the parallel trends assumption in Assumption~\ref{CIAdY} collapses to CIA-$D$, i.e.\ selection-on-observables formulated in terms of levels rather than differences of potential outcomes (see \citet{KnausPfleiderer2026} for a general graphical treatment of such conditioning-set questions).\footnote{%
		To see this, write $\bar{Y}_{T-1}$ explicitly in the conditioning set as a slight abuse of our previous notation:
		\begin{align*}
			&\nabla Y_T(d) \perp D_T \mid \bar{D}_{T-1},\, \bar{X}_{T-1},\, \bar{Y}_{T-1} \\
			\iff\;&  Y_T(d) - Y_{T-1}(d) \perp D_T \mid \bar{D}_{T-1},\, \bar{X}_{T-1},\, \bar{Y}_{T-1} \\
			\iff\;&  Y_T(d) - Y_{T-1} \perp D_T \mid \bar{D}_{T-1},\, \bar{X}_{T-1},\, \bar{Y}_{T-1} \\
			\iff\;&  Y_T(d) \perp D_T \mid \bar{D}_{T-1},\, \bar{X}_{T-1},\, \bar{Y}_{T-1},
		\end{align*}
		where the second equivalence follows from the absence of anticipation effects, so that $Y_{T-1}(d) = Y_{T-1}$ for all $d$, and the third from the fact that $Y_{T-1}$ is measurable with respect to the conditioning set and therefore fixed conditional on it, so that subtracting it does not affect conditional independence.}
	
	We note that in the DiD literature, typically weaker versions of parallel trends than Assumption~\ref{CIAdY} are imposed. Most commonly, parallel trends are required to hold only on average and only for the potential outcome difference under nontreatment, $\nabla Y_T(0)$, conditional on not having been treated in earlier periods, $\bar{D}_{T-1}=0$. Together with suitable common support conditions, this is sufficient for identifying the average treatment effect on the treated (ATET) among units switching from nontreatment in earlier periods to treatment in period $T$; see, for example, \citet{Abadie2005} for semiparametric ATET identification and the literature on staggered treatment adoption, including \citet{borusyak2024revisiting}, \citet{CallawaySantAnna2018}, \citet{GoodmanBacon2018}, \citet{cha19}, and \citet{sun2021estimating}. For conceptual clarity, we focus here on the stronger Assumption~\ref{CIAdY}, which is imposed on the entire distribution of potential outcomes rather than only on means, and for all treatment levels $d$ rather than only for $d=0$, as also considered in \citet{Fricke2017}, \citet{deChaisemartin2022did}, and \citet{haddad2024difference}. This stronger assumption also allows identification of effects in the total population, such as the ATE, rather than only the ATET. For the sake of brevity, we also omit an explicit discussion of common support conditions.
	
	We first show that, under Model A, all three conditional independence assumptions hold simultaneously. Recall from equation~\eqref{treatsimp} that
	\begin{align}
		D_T = D_{T-1} + \mathcal{F}_D(T,\bar{X}_{T-1}) - \mathcal{F}_D(T-1,\bar{X}_{T-2}) + \varepsilon_T.
	\end{align}
	Conditional on $\bar{X}_{T-1}$ and $D_{T-1}$, the treatment level $D_T$ is therefore a one-to-one function of the shock $\varepsilon_T$. In the same way, the treatment change $\nabla D_T$ is a one-to-one function of $\varepsilon_T$ conditional on $\bar X_{T-1}$. The conditioning sets of CIA-$D$ and CIA-$\nabla Y$ additionally contain the treatment history $\bar D_{T-1}$, and the conditioning set of CIA-$D$ also contains the outcome history $\bar Y_{T-1}$. Given $\bar X_{T-1}$, both histories are functions of $(U,\bar V_{T-1},\bar W_{T-1})$. This follows from the treatment equation~\eqref{restrictedmodel} for $\bar D_{T-1}$, and from the outcome equation in~\eqref{strucmodels}, applied to periods up to $T-1$, for $\bar Y_{T-1}$. Condition~\eqref{condind1} makes $\varepsilon_T$ independent of $(U,\bar V_T,\bar W_{T-1})$ given $\bar X_{T-1}$. By the weak union property of conditional independence, this independence is preserved when functions of these variables are added to the conditioning set. Therefore,
	\begin{align}\label{condind_levels}
		\{U,\bar V_T,\bar W_{T-1}\} \perp \varepsilon_T \mid \bar X_{T-1},\bar D_{T-1},\bar Y_{T-1},
	\end{align}
	and the same independence holds with the conditioning set reduced to $(\bar X_{T-1},\bar D_{T-1})$.
	
	Since $D_T$ is a one-to-one function of $\varepsilon_T$ given $\bar X_{T-1}$ and $D_{T-1}$, condition~\eqref{condind_levels} implies that $\{U,\bar V_T\}$ is independent of $D_T$ given $(\bar X_{T-1},\bar D_{T-1},\bar Y_{T-1})$, as well as given $(\bar X_{T-1},\bar D_{T-1})$. For CIA-$D$, the potential outcome $Y_T(d)=\mathcal{F}_Y(T,\bar D_{T-1},d,\bar X_{T-1},\bar V_T,U)$ is a function of $(U,\bar V_T)$ given $(\bar X_{T-1},\bar D_{T-1},\bar Y_{T-1})$. It follows that $Y_T(d)\perp D_T\mid\bar D_{T-1},\bar X_{T-1},\bar Y_{T-1}$, which is CIA-$D$. For CIA-$\nabla Y$, the potential outcome change is $\nabla Y_T(d)=Y_T(d)-Y_{T-1}$, with $Y_{T-1}=\mathcal{F}_Y(T-1,\bar D_{T-1},\bar X_{T-2},\bar V_{T-1},U)$. Given $\bar D_{T-1}$ and $\bar X_{T-1}$, $\nabla Y_T(d)$ is therefore a function of $(U,\bar V_T)$, which yields CIA-$\nabla Y$. The fixed effect $U$ may remain in $\nabla Y_T(d)$, for instance when it enters the outcome equation with a time-varying coefficient. This is admissible, since the argument only requires $(U,\bar V_T)$ to be independent of $D_T$ given the conditioning set. CIA-$\nabla D$ holds by Proposition~\ref{prop:modelA}. These results are summarized in Proposition~\ref{modelA_all}.
	
	\begin{proposition}(All three identification strategies valid under Model A).\\
		\label{modelA_all}
		Under the treatment equation~\eqref{restrictedmodel}, the random walk~\eqref{randomwalk}, the unrestricted outcome equation in~\eqref{strucmodels}, and condition~\eqref{condind1},
		\begin{align}
			\{U,\bar V_T\} \perp D_T \mid \bar X_{T-1},\bar D_{T-1},\bar Y_{T-1}
			\quad\text{and}\quad
			\{U,\bar V_T\} \perp D_T \mid \bar X_{T-1},\bar D_{T-1}.
		\end{align}
		Consequently, Assumptions~\ref{CIAdD}, \ref{CIAD}, and~\ref{CIAdY} all hold simultaneously.
	\end{proposition}
	
	Proposition~\ref{modelA_all} shows that under Model A, all three strategies achieve identification by exploiting the same source of variation, namely the conditionally random treatment shock $\varepsilon_T$, the innovation of the random walk $W_T$, which is independent of the fixed effect $U$. This equivalence is the basis for the overidentification tests developed in Section~\ref{sec:overid}.
	
	Having established that CIA-$\nabla D$, CIA-$D$, and CIA-$\nabla Y$ coincide under Model A, we now ask whether they are otherwise non-nested. Given the results of Section~\ref{strucmod}, however, the only structural route to Assumption~\ref{CIAdD} that we have identified in the heterogeneous, nonparametric case is Model A itself, under which the other two assumptions hold as well. We are not aware of a structural model strictly weaker than Model A that delivers Assumption~\ref{CIAdD} for a heterogeneous treatment effect model while CIA-$D$ or CIA-$\nabla Y$ fails. We instead move from this nonparametric comparison to the partially linear model that underlies Model B's instrumental-variables route, where the conditional average marginal effects targeted by the different strategies coincide and can be compared directly. What we can establish there, and what is in fact the relevant comparison for the linear and instrumental-variables results of Section~\ref{strucmod} and for the simulation study of Section~\ref{sim}, is non-nesting at the level of the linear moment conditions that the corresponding ordinary least squares (OLS) estimators rely on. Consider the partially linear outcome model~\eqref{partiallylinear} with a treatment effect $\delta_T(\bar X_{T-1})$ that is homogeneous given covariates, and define the three moment conditions
	\begin{align}
		\text{(IV)}&\qquad \mathrm{Cov}(\nabla D_T,\eta_T\mid \bar X_{T-1}) = 0, \notag\\
		\text{(D)}&\qquad \mathrm{Cov}(D_T,\eta_T\mid \bar D_{T-1},\bar X_{T-1},\bar Y_{T-1}) = 0, \notag\\
		\text{(}\nabla\text{Y)}&\qquad \mathrm{Cov}(D_T,\nabla\eta_T\mid \bar D_{T-1},\bar X_{T-1}) = 0, \quad \nabla\eta_T:=\eta_T-\eta_{T-1}, \notag
	\end{align}
	corresponding to the exogeneity conditions required for the OLS estimators based on, respectively, the instrumental-variables route of Proposition~\ref{prop:modelBiv}, the level-based selection-on-observables regression underlying CIA-$D$, and the DiD regression underlying CIA-$\nabla Y$. These moment conditions are precisely what is being verified, numerically, in the simulation study of Section~\ref{sim}.
	
	\begin{proposition}[Non-nesting of the linear moment conditions]\label{prop:nonnest_moments}
		Conditions (IV), (D), and ($\nabla$Y) are pairwise non-nested: none implies any of the others in general.
	\end{proposition}
	
	We demonstrate Proposition~\ref{prop:nonnest_moments} through three examples, one for each pair, each establishing both directions, all built on the base case $D_0=U+W_0$, $D_1=U+W_1$ with $U,W_0,W_1,V_0,V_1$ mutually independent.
	
	\begin{example}[(IV) vs.\ (D)]\label{ex:iv_vs_d}
		Let $\eta_0=V_0$, so that $Y_0=\delta D_0+V_0$ carries no information on $U$ beyond $D_0$, and conditioning on $(D_0,Y_0)$ in (D) yields the same conditional covariances as conditioning on $D_0$ alone. Let further $\eta_1=\gamma U+V_1$ for $\gamma\neq0$. Then (IV) holds, since $\mathrm{Cov}(\nabla D_1,\eta_1)=\mathrm{Cov}(W_1-W_0,\gamma U+V_1)=0$ by mutual independence, but (D) fails, since $\mathrm{Cov}(D_1,\eta_1\mid D_0) = \gamma\,\mathrm{Var}(U\mid D_0) \neq 0$ whenever $\gamma\neq0$ and $\mathrm{Var}(U),\mathrm{Var}(W_0)>0$. Conversely, let $\eta_1=U+W_0+V_1=D_0+V_1$, so that $U$ and $W_0$ affect the period-1 outcome only through the lagged treatment $D_0$, for instance through a dynamic treatment effect. Then (D) holds, since $\mathrm{Cov}(D_1,\eta_1\mid D_0)=\mathrm{Cov}(D_1,V_1\mid D_0)=0$, but (IV) fails, since $\mathrm{Cov}(\nabla D_1,\eta_1)=\mathrm{Cov}(W_1-W_0,U+W_0+V_1)=-\mathrm{Var}(W_0)\neq0$.
	\end{example}
	
	\begin{example}[(D) vs.\ ($\nabla$Y)]\label{ex:d_vs_dy}
		Let $\eta_0=U+V_0$ and $\eta_1=U+V_1$ (a common, time-invariant loading on $U$). Then $\nabla\eta_1=V_1-V_0$, so ($\nabla$Y) holds trivially, since $\mathrm{Cov}(D_1,V_1-V_0\mid D_0)=0$. However, conditioning on $(D_0,Y_0)$ does not fully reveal $U$ when $\mathrm{Var}(V_0)>0$, since $Y_0$ and $D_0$ together solve only two linear equations in the three unobservables $(U,W_0,V_0)$; hence $\mathrm{Cov}(D_1,\eta_1\mid D_0,Y_0) = \mathrm{Var}(U\mid D_0,Y_0) \neq 0$ generically, so (D) fails. Conversely, let $\eta_0=U+V_0$ and $\eta_1=\rho_\eta\eta_0+V_1$ with $\rho_\eta\neq1$, so that the outcome residual follows an autoregressive process of order one. Then $\eta_1=\rho_\eta(Y_0-\delta D_0)+V_1$, so that conditioning on $(D_0,Y_0)$ fixes $\eta_0$ and gives $\mathrm{Cov}(D_1,\eta_1\mid D_0,Y_0) = \mathrm{Cov}(D_1,V_1\mid D_0,Y_0) = 0$, so (D) holds. But $\nabla\eta_1=(\rho_\eta-1)\eta_0+V_1=(\rho_\eta-1)(U+V_0)+V_1$, and since the conditioning set for ($\nabla$Y) excludes $Y_0$, $\mathrm{Cov}(D_1,\nabla\eta_1\mid D_0) = (\rho_\eta-1)\,\mathrm{Cov}(U+W_1,U+V_0\mid D_0)=(\rho_\eta-1)\,\mathrm{Var}(U\mid D_0)\neq0$, so ($\nabla$Y) fails.
	\end{example}
	
	\begin{example}[(IV) vs.\ ($\nabla$Y)]\label{ex:iv_vs_dy}
		Let $\eta_0=\gamma_0 U+V_0$ and $\eta_1=\gamma_1 U+V_1$ with $\gamma_1\neq\gamma_0$. Then (IV) holds regardless of $\gamma_1$, since $\mathrm{Cov}(\nabla D_1,\eta_1)=\mathrm{Cov}(W_1-W_0,\gamma_1 U+V_1)=0$ by mutual independence. But $\nabla\eta_1=(\gamma_1-\gamma_0)U+V_1-V_0$, so $\mathrm{Cov}(D_1,\nabla\eta_1\mid D_0) = (\gamma_1-\gamma_0)\,\mathrm{Var}(U\mid D_0) \neq 0$, so ($\nabla$Y) fails. Conversely, let $\eta_0=U$ and $\eta_1=-W_0+V_1$. Direct calculation, using $\mathrm{Cov}(U,W_0\mid D_0)=-\mathrm{Var}(U\mid D_0)$, gives $\mathrm{Cov}(D_1,\nabla\eta_1\mid D_0)=\mathrm{Cov}(D_1,-W_0+V_1-U\mid D_0)=\mathrm{Var}(U\mid D_0)-\mathrm{Var}(U\mid D_0)=0$, so ($\nabla$Y) holds, while $\mathrm{Cov}(\nabla D_1,\eta_1)=\mathrm{Cov}(W_1-W_0,-W_0+V_1)=\mathrm{Var}(W_0)\neq0$, so (IV) fails.
	\end{example}
	
	Examples~\ref{ex:iv_vs_d}--\ref{ex:iv_vs_dy} establish Proposition~\ref{prop:nonnest_moments}. This non-nesting result is the operative one for the remainder of the paper: it explains why, in the simulation study of Section~\ref{sim}, each linear estimator can be biased while the others remain consistent, and it underlies the construction of the overidentification tests in Section~\ref{sec:overid}, which compare estimates that are guaranteed to coincide only under the stronger structural assumption of Model A.
	\section{Overidentification Tests and Double Robustness}
	\label{sec:overid}
	
	The equivalence result of Proposition~\ref{modelA_all} implies that when Model A holds, the resulting causal parameters coincide in the population, and a Hausman-style comparison of estimates constitutes an overidentification test. Before developing these tests, we establish that all three identification strategies --- CIA-$\nabla D$, CIA-$D$, and CIA-$\nabla Y$ --- can in principle be used to identify the same treatment-change-based causal parameters, as defined in equations~\eqref{eff1} and~\eqref{eff2}, which puts them on a common footing for comparison.
	
	Under CIA-$\nabla D$ (Assumption~\ref{CIAdD}) and common support (Assumption~\ref{comsup}), the average potential outcome under a treatment change $\nabla d$ is identified as
	\begin{align}\label{id_nablad}
		E[Y_t(D_{t-1}+\nabla d) \mid \bar{X}_{t-1}]
		= E\bigl[ E[Y_t \mid \nabla D_t = \nabla d,\, \bar{X}_{t-1}] \mid \bar{X}_{t-1} \bigr],
	\end{align}
	where identification follows from the conditional independence in Assumption~\ref{CIAdD} by the same argument as in Section~\ref{setup}. Note that $\bar{X}_{t-1}$ may include lagged outcomes as control variables under CIA-$\nabla D$.
	
	Under CIA-$D$ (Assumption~\ref{CIAD}), setting $d = D_{t-1} + \nabla d$ and applying the conditional independence directly gives
	\begin{align}\label{id_d}
		E[Y_t(D_{t-1}+\nabla d) \mid \bar{D}_{t-1},\, \bar{X}_{t-1}]
		= E\bigl[Y_t \mid D_t = D_{t-1}+\nabla d,\, \bar{D}_{t-1},\, \bar{X}_{t-1}\bigr],
	\end{align}
	where $\bar{X}_{t-1}$ under CIA-$D$ may include lagged outcomes $\bar{Y}_{t-1}$. Averaging over the distribution of $\bar{D}_{t-1}$ conditional on $\bar{X}_{t-1}$ then identifies $E[Y_t(D_{t-1}+\nabla d) \mid \bar{X}_{t-1}]$.
	
	Under CIA-$\nabla Y$ (Assumption~\ref{CIAdY}), the parallel trends assumption identifies the average potential outcome change under treatment level $d = D_{t-1}+\nabla d$ as
	\begin{align}\label{id_nablay}
		E[\nabla Y_t(D_{t-1}+\nabla d) \mid \bar{D}_{t-1}, \bar{X}_{t-1}]
		= E[\nabla Y_t \mid D_t = D_{t-1}+\nabla d,\, \bar{D}_{t-1},\, \bar{X}_{t-1}].
	\end{align}
	Adding back the observed lagged outcome $E[Y_{t-1} \mid \bar{D}_{t-1}, \bar{X}_{t-1}]$ and averaging over the distribution of $\bar{D}_{t-1}$ conditional on $\bar{X}_{t-1}$ identifies $E[Y_t(D_{t-1}+\nabla d) \mid \bar{X}_{t-1}]$. As discussed in Section~\ref{sec:comparison}, the conditioning set $\bar{X}_{t-1}$ must not include lagged outcomes $\bar{Y}_{t-1}$ under CIA-$\nabla Y$, since their inclusion collapses this assumption to CIA-$D$.
	
	Since all three strategies identify $E[Y_t(D_{t-1}+\nabla d) \mid \bar{X}_{t-1}]$ under their respective assumptions, the average treatment effects $\Delta_t(\nabla d)$ and $\Delta_t(\nabla d, \nabla d')$ are also identified under each strategy, and the resulting estimates should agree whenever Model A holds, so that all three assumptions are simultaneously satisfied. This motivates Hausman-type overidentification tests: under Model A, the estimates of $\Delta_t(\nabla d, \nabla d')$, denoted by $\hat{\Delta}_{t,j}$ and $\hat{\Delta}_{t,k}$ for strategies $j,k\in\{\nabla D, D, \nabla Y\}$, are asymptotically equivalent. Furthermore, the test statistic
	\begin{align}\label{eq:hausman}
		H_{t,jk} = \frac{(\hat{\Delta}_{t,j} - \hat{\Delta}_{t,k})^2}{\hat{\sigma}^2_{t,jk}}
		\;\overset{H_0}{\longrightarrow}\; \chi^2(1),
	\end{align}
	where $\hat{\sigma}^2_{t,jk}$ is a consistent estimate of the variance of $\hat{\Delta}_{t,j} - \hat{\Delta}_{t,k}$, is asymptotically valid under standard regularity conditions. Since both estimates are computed from the same sample, this variance equals $\sigma^2_{t,j}+\sigma^2_{t,k}-2\sigma_{t,jk}$, where $\sigma^2_{t,j}$ and $\sigma^2_{t,k}$ denote the variances of the two estimates and $\sigma_{t,jk}$ their covariance. It can be estimated by a bootstrap that re-estimates both strategies in each bootstrap sample, or from the influence functions of the two estimators. Rejection of $H_{t,jk}$ establishes that Model A, or one of the maintained assumptions, fails, but cannot identify which one without additional prior restrictions. Non-rejection is consistent with Model A but also with violations that bias both strategies in exactly the same way, a case that may be ruled out by imposing a causal faithfulness assumption, see e.g.\ \citet{Pearl00}.
	
	When only Model B's instrumental-variables route (Proposition~\ref{prop:modelBiv}) is available rather than the full Model A, the treatment effect $\delta_T(\bar X_{T-1})$ is not identified directly by the OLS regression of $Y_T$ on $\nabla D_T$: as the expanded form of equation~\eqref{eq:ivratio} already makes explicit, this regression instead identifies
	\begin{align}\label{eq:ivD_estimand}
		\beta_{\text{IV-}\nabla D}(\bar x) := \frac{\mathrm{Cov}(\nabla D_T,Y_T\mid\bar x)}{\mathrm{Var}(\nabla D_T\mid\bar x)} = \delta_T(\bar x)\cdot\lambda_T(\bar x)
	\end{align}
	within each covariate cell. Recovering $\delta_T(\bar x)$ from this regression therefore requires dividing by an estimate of the persistence factor $\lambda_T(\bar x)$, exactly as in the standard instrumental-variables ratio of equation~\eqref{eq:ivratio}. Under Model A, $\lambda_T(\bar x)=1$ for all $\bar x$, and the rescaling becomes unnecessary, consistent with the fact that under Model A the regression of $Y_T$ on $\nabla D_T$ already estimates $\delta_T(\bar x)$ directly. The variance of the persistence-factor estimate must be propagated into the standard error of the rescaled estimate, for instance via the delta method. We return to this in the empirical application of Section~\ref{app}.
	
	The non-nesting of (IV) and ($\nabla$Y) established in Proposition~\ref{prop:nonnest_moments} carries a further implication that complements the overidentification tests, concerning the popular two-way fixed effects (TWFE) approach, which simultaneously relies on differenced outcomes $\nabla Y$ and treatment changes $\nabla D$. Under specific modeling assumptions, TWFE bears a double robustness (DR) property for a treatment effect $\delta_T(\bar X_{T-1})$ that takes the same value in periods $T-1$ and $T$: the linear projection coefficient in a population regression of $\nabla Y_T$ on $\nabla D_T$ conditional on $\bar X_{T-1}$, denoted by $\beta_{\text{TWFE}}(\bar X_{T-1})$,
	recovers the treatment effect within each covariate cell, if the time-constant unobserved factor $U$ is additively separable in either the treatment equation or the outcome equation, without requiring both simultaneously. As in Proposition~\ref{prop:modelBiv}, this does not rule out the treatment effect varying across covariate values, but in contrast to the assumptions underlying Proposition~\ref{prop:modelBiv}, treatment effects must not change between $T-1$ and $T$ for the realized covariate path of a given unit. 
	
	
	To formalize the DR result, we consider two distinct structural models, both restricted to the partially linear outcome equation~\eqref{partiallylinear} with a treatment effect that is homogeneous across periods $T-1$ and $T$ given the covariate history, $\delta_{T-1}(\bar X_{T-2})=\delta_{T}(\bar X_{T-1})=\delta(\bar X_{T-1})$. This requires the treatment effect to vary with the covariate history only through $\bar X_{T-2}$. The first model corresponds to Model B, under which $U$ is additively separable in the treatment equation~\eqref{restrictedmodel0} and $\eta_T=f(U,\bar X_{T-1},V_T)$ as in the no-dynamic-effects restriction~\eqref{restrictedmodel2} underlying Proposition~\ref{prop:modelBiv}, which delivers the exogeneity condition $\mathrm{Cov}(\nabla D_T,\eta_T\mid\bar X_{T-1})=0$. The second model departs from this by relaxing additive separability of $U$ in the treatment equation, while instead imposing it in the outcome equation. Specifically, consider a treatment equation that does not require $U$ to be additively separable,
	\begin{align}\label{restrictedmodel0gen}
		D_T = \mathcal{F}_D\left(T, \bar{X}_{T-1}, \bar{W}_T, U\right),
	\end{align}
	and an outcome equation with additively separable $U$,
	\begin{align}\label{restrictedmodel2add}
		Y_T = \delta(\bar X_{T-1}) D_T + g\left(\bar{X}_{T-1}, V_T\right) + U,
	\end{align}
	which splits $U$ out additively from the residual, while maintaining the absence of dynamic effects. Under~\eqref{restrictedmodel2add}, $U$ cancels exactly in $\nabla Y_T(\nabla d) = \delta(\bar X_{T-1})\,\nabla d + g(\bar X_{T-1},V_T) - g(\bar X_{T-2},V_{T-1})$, regardless of how $U$ enters the treatment equation~\eqref{restrictedmodel0gen}, so that the relevant residual $\nabla\tilde\eta_T := g(\bar X_{T-1},V_T) - g(\bar X_{T-2},V_{T-1})$ depends only on $(V_T,V_{T-1})$ given $\bar X_{T-1}$. Under~\eqref{restrictedmodel0gen}, however, $\nabla D_T$ may depend on $U$. Exogeneity of $\nabla D_T$ with respect to $\nabla\tilde\eta_T$ therefore requires the time-varying outcome unobservables in both periods to be independent of $U$ and $\bar W_T$, i.e.\ $V_t \perp \{U,\bar{W}_T\} \mid \bar{X}_{T-1}$ for $t\in\{T-1,T\}$. This rules out that $V_T$ or $V_{T-1}$ carries information about $U$, which would otherwise correlate $\nabla\tilde\eta_T$ with $\nabla D_T$ through the dependence of $\nabla D_T$ on $U$. For period $T-1$, this condition implies the following condition on the lagged outcome shock:
	\begin{align}\label{nofeedback}
		V_{T-1} \perp \bar{W}_T \mid \bar{X}_{T-1},
	\end{align}
	which, together with condition~\eqref{condind3}, implies strict exogeneity of the treatment shocks from time-varying unobservables affecting the outcome conditional on $\bar{X}_{T-1}$ for periods $T$ and $T-1$, as typically assumed in panel models. Condition~\eqref{nofeedback} is required because TWFE considers the differenced outcome $\nabla Y_T = Y_T - Y_{T-1}$ as dependent variable, thus including $Y_{T-1}$, which depends on $V_{T-1}$. Condition~\eqref{condind3} alone, which restricts only $V_T$, says nothing about $V_{T-1}$. A condition on the lagged residual is also required under Model B, as the exogeneity condition $\mathrm{Cov}(\nabla D_T,\eta_T\mid\bar X_{T-1})=0$ established for period $T$ in Proposition~\ref{prop:modelBiv} does not by itself extend to the lagged residual $\eta_{T-1}$. Since $\eta_{T-1}=f_{T-1}(U,\bar X_{T-2},V_{T-1})$ depends jointly on $U$ and $V_{T-1}$, condition~\eqref{nofeedback} must then be strengthened to $\{U,V_{T-1}\}\perp\bar W_T\mid\bar X_{T-1}$, the counterpart of condition~\eqref{condind3} for period $T-1$. Imposing the independence of $U$ and of $V_{T-1}$ from $\bar W_T$ separately is insufficient when $f_{T-1}$ is nonlinear in $(U,V_{T-1})$. Without such a condition, the TWFE projection coefficient would generally differ from $\delta(\bar X_{T-1})$ even when Model B holds for period $T$. Proposition~\ref{prop:doublerobust} formally states the DR result, with route-specific conditions that each imply condition~\eqref{nofeedback}.
	
	\begin{proposition}(Double robustness of the TWFE moment condition for $\delta(\bar X_{T-1})$).\\
		\label{prop:doublerobust}
		Suppose $Y_{T-1}=\delta(\bar X_{T-1}) D_{T-1}+\eta_{T-1}$ and $Y_T=\delta(\bar X_{T-1}) D_T+\eta_T$ with a common $\delta(\bar X_{T-1})$ in periods $T-1$ and $T$, with $\nabla\eta_T := \eta_T-\eta_{T-1}$ and $\mathrm{Var}(\nabla D_T\mid\bar X_{T-1})>0$. Under either of the following two sets of conditions, $\mathrm{Cov}(\nabla D_T,\nabla\eta_T \mid \bar{X}_{T-1}) = 0$, so that $\beta_{\text{TWFE}}(\bar X_{T-1})=\delta(\bar X_{T-1})$:
		\begin{enumerate}[(i)]
			\item Treatment equation~\eqref{restrictedmodel0} with $U$ additively separable, $\eta_t=f_t(U,\bar X_{t-1},V_t)$ for $t\in\{T-1,T\}$ as in~\eqref{restrictedmodel2}, and $\{U,V_t\}\perp\bar W_T\mid\bar X_{T-1}$ for $t\in\{T-1,T\}$, i.e.\ condition~\eqref{condind3} together with its counterpart for $V_{T-1}$.
			\item Treatment equation~\eqref{restrictedmodel0gen}, outcome equation~\eqref{restrictedmodel2add} in periods $T-1$ and $T$ with $U$ additively separable, and $V_t\perp\{U,\bar W_T\}\mid\bar X_{T-1}$ for $t\in\{T-1,T\}$.
		\end{enumerate}
		Each set of conditions implies condition~\eqref{nofeedback}.
	\end{proposition}
	
	\begin{proof}
		Under route $(i)$, conditional on $\bar X_{T-1}$, $\nabla D_T$ is a function of $\bar W_T$ alone by equation~\eqref{treatsimp2}. Furthermore, $\eta_T=f_T(U,\bar X_{T-1},V_T)$ is a function of $(U,V_T)$, and $\eta_{T-1}=f_{T-1}(U,\bar X_{T-2},V_{T-1})$ is a function of $(U,V_{T-1})$, since $\bar X_{T-2}$ is contained in $\bar X_{T-1}$. The condition $\{U,V_T\}\perp\bar W_T\mid\bar X_{T-1}$ therefore gives $\mathrm{Cov}(\nabla D_T,\eta_T\mid\bar X_{T-1})=0$, and the condition $\{U,V_{T-1}\}\perp\bar W_T\mid\bar X_{T-1}$ gives $\mathrm{Cov}(\nabla D_T,\eta_{T-1}\mid\bar X_{T-1})=0$. Both conditions are imposed jointly on $U$ and $V_t$ because $f_t$ may be nonlinear in $(U,V_t)$. By bilinearity of the covariance, $\mathrm{Cov}(\nabla D_T,\nabla\eta_T\mid\bar X_{T-1})=0$.
		
		Under route $(ii)$, $\eta_t=g(\bar X_{t-1},V_t)+U$ for $t\in\{T-1,T\}$, so that $\nabla\eta_T = g(\bar X_{T-1},V_T)-g(\bar X_{T-2},V_{T-1})=\nabla\tilde\eta_T$ and $U$ cancels exactly. Conditional on $\bar X_{T-1}$, $\nabla D_T$ is a function of $(U,\bar W_T)$ under the treatment equation~\eqref{restrictedmodel0gen}. The two terms of $\nabla\eta_T$ are functions of $V_T$ and $V_{T-1}$, respectively. The condition $V_t\perp\{U,\bar W_T\}\mid\bar X_{T-1}$ for $t\in\{T-1,T\}$ therefore gives a zero conditional covariance of $\nabla D_T$ with each of the two terms, and bilinearity of the covariance gives $\mathrm{Cov}(\nabla D_T,\nabla\eta_T\mid\bar X_{T-1})=0$. The independence of $V_t$ from $U$ is required because $\nabla D_T$ depends on $U$ under~\eqref{restrictedmodel0gen}. The cancellation of $U$ in the outcome equation removes the explicit $U$ term from $\nabla\eta_T$, and the independence condition removes any statistical dependence on $U$ that $V_T$ or $V_{T-1}$ could carry.
		
		In both routes, $\nabla Y_T=\delta(\bar X_{T-1})\nabla D_T+\nabla\eta_T$, so that $\mathrm{Cov}(\nabla D_T,\nabla Y_T\mid\bar X_{T-1}) = \delta(\bar X_{T-1})\,\mathrm{Var}(\nabla D_T\mid\bar X_{T-1}) + \mathrm{Cov}(\nabla D_T,\nabla\eta_T\mid\bar X_{T-1}) = \delta(\bar X_{T-1})\,\mathrm{Var}(\nabla D_T\mid\bar X_{T-1})$, and hence $\beta_{\text{TWFE}}(\bar X_{T-1}) = \delta(\bar X_{T-1})$.
	\end{proof}
	
	The two routes in Proposition~\ref{prop:doublerobust} are non-nested, which constitutes the DR property of TWFE for the linear projection coefficient: additive separability of $U$ in the treatment equation~\eqref{restrictedmodel0} neither implies nor is implied by its additive separability in the outcome equation~\eqref{restrictedmodel2add}. Under route (ii), $U$ need not cancel from $\nabla D_T$; under route (i), $\eta_T$ need not take the additively separable form of~\eqref{restrictedmodel2add}. In both cases, what matters is that $U$ enters the relevant equation with the \emph{same} coefficient in periods $T-1$ and $T$, so that it cancels exactly under differencing. In addition, the unobservables entering the outcome residuals in periods $T-1$ and $T$ must be independent of the unobservables that drive $\nabla D_T$ given $\bar X_{T-1}$. Under route (i), these are $\bar W_T$, and the condition is imposed jointly on $U$ and $V_t$. Under route (ii), $\nabla D_T$ also depends on $U$, so that $V_t$ must be independent of $(U,\bar W_T)$. This structural double robustness is related to, but distinct from, the findings of \citet{arkhangelsky2022doubly} and \citet{arkhangelsky2021double}. Those papers distinguish a model-based path restricting the outcome equation and a design-based path restricting the treatment equation, and construct estimators that are consistent if either restriction holds. Our result differs in that our framework accommodates continuously varying rather than binary treatments, and that both routes here concern the linear projection of the outcome on the treatment, rather than a nonparametric conditional mean.
	
	By Proposition~\ref{prop:doublerobust}, the treatment effect may be heterogeneous across covariates, but not between the two time periods being differenced. The following example shows that this restriction across time cannot be dispensed with. \citet*{arganaraz2025randomly} make a related observation in the context of first-difference regressions, where a treatment effect changing over time likewise introduces a dependence on the baseline treatment. With a treatment effect that changes over time, the outcome change of a unit contains the change in the treatment effect at its prior treatment level. In designs with a binary treatment adopted at different times, this concerns units treated earlier, which serve as controls for units treated later, so that changes in their treatment effects enter the TWFE estimate with negative weights, as documented by \citet{GoodmanBacon2018} and \citet{cha19}.
	\begin{example}[Failure of double robustness under a time-varying treatment effect]\label{ex:dr_fails}
		Let $U,W_0,W_1,V_0,V_1$ be mutually independent with $\mathrm{Var}(W_0)=\mathrm{Var}(W_1)=\sigma_W^2>0$, and define
		\begin{align}
			D_0=U+W_0,\quad D_1=U+W_1,\quad Y_0=D_0+U+V_0,\quad Y_1(d)=U+V_1\ \text{for all }d.
		\end{align}
		There are no dynamic treatment effects, $U$ is additively separable in both the treatment and the outcome equations, and the time-varying outcome shocks are independent of the treatment shocks. For every $a$, $Y_1(D_0+a)=U+V_1 \perp \nabla D_1=W_1-W_0$, so Assumption~\ref{CIAdD} holds. For every fixed current treatment level $d$, $\nabla Y_1(d) = Y_1(d) - Y_0 = -D_0+V_1-V_0$, which, conditional on $D_0$, is independent of $D_1$, so Assumption~\ref{CIAdY} also holds. Nevertheless,
		\begin{align}
			\nabla Y_1(a) = Y_1(D_0+a)-Y_0 = -D_0+V_1-V_0,
		\end{align}
		so that $\mathrm{Cov}(\nabla D_1,\nabla Y_1(a)) = -\mathrm{Cov}(W_1-W_0,D_0) = \mathrm{Var}(W_0) = \sigma_W^2 > 0$. The conditional-independence claim $\nabla Y_T(\nabla d) \perp \nabla D_T \mid \bar X_{T-1}$ therefore fails even though both Assumption~\ref{CIAdD} and Assumption~\ref{CIAdY} hold individually. The observed linear projection coefficient in this example is $\mathrm{Cov}(\nabla D_1,\nabla Y_1)/\mathrm{Var}(\nabla D_1) = \sigma_W^2/(2\sigma_W^2) = 1/2$, whereas the contemporaneous treatment effect in period 1 is zero, thus violating the homogeneity assumption of Proposition~\ref{prop:doublerobust}.
	\end{example}
	
	Finally, it is worth discussing why no IV-based ratio, of the kind used in Proposition~\ref{prop:modelBiv}, is needed for TWFE-based identification when the treatment effect is homogeneous (or linear) in the treatment level with a common slope $\delta(\bar X_{T-1})$ in both periods. Under homogeneity, $\nabla Y_T = \delta(\bar X_{T-1}) D_T + \eta_T - (\delta(\bar X_{T-1}) D_{T-1}+\eta_{T-1}) = \delta(\bar X_{T-1})\,\nabla D_T + \nabla\eta_T$: the prior level $D_{T-1}$ cancels algebraically from the regression equation itself, making the instrumental variable approach obsolete. This is different from the level regression of Proposition~\ref{prop:modelBiv}: $D_{T-1}$ remains in the equation $Y_T = \delta_T(\bar X_{T-1}) D_{T-1} + \delta_T(\bar X_{T-1})\,\nabla D_T + \eta_T$ and its correlation with $\nabla D_T$ is absorbed, rather than eliminated, by dividing by the first-stage covariance $\mathrm{Cov}(\nabla D_T,D_T\mid\bar X_{T-1})$. Once the response function relating the outcome to the treatment level is nonlinear, so that the treatment effect is heterogeneous in the treatment level, this algebraic cancellation no longer occurs for TWFE. For a nonlinear response function $m$, $Y_T(D_T)-Y_{T-1}(D_{T-1}) = m(D_T)-m(D_{T-1})$ depends on both endpoints $D_{T-1}$ and $D_T$ rather than on $\nabla D_T$ alone, so that the TWFE coefficient is a weighted average of the marginal effects over the treatment levels between $D_{T-1}$ and $D_T$, with weights that can be negative, as discussed in Section~\ref{strucmod}. The DR result of Proposition~\ref{prop:doublerobust} is therefore specific to effects that are linear in treatment levels.
	\section{Simulation Study}\label{sim}
	
	This section presents two simulation panels illustrating the non-nesting results of Section~\ref{sec:comparison}, the rescaling result of Section~\ref{sec:overid}, and the double robustness (DR) of TWFE established in Proposition~\ref{prop:doublerobust}. Throughout, all four estimators --- the OLS regression of $Y_1$ on $\nabla D_1$, the OLS regression of $\nabla Y_1$ on $D_1$ and $D_0$, the OLS regression of $Y_1$ on $D_1$, $D_0$, and $Y_0$, and TWFE --- are linear estimators. The first three rely on the moment conditions (IV), ($\nabla$Y), and (D) of Section~\ref{sec:comparison}, respectively, and TWFE relies on the condition $\mathrm{Cov}(\nabla D_1,\nabla\eta_1)=0$ of Proposition~\ref{prop:doublerobust}. We label the first of these estimators ``IV-$\nabla D$'' throughout this section, rather than ``CIA-$\nabla D$'', to reflect that what the simulation verifies is the linear moment condition~\eqref{eq:ivexog} of Proposition~\ref{prop:modelBiv} for instrument-based effect identification, not the full nonparametric conditional independence of Assumption~\ref{CIAdD}. Since the data-generating processes below have a single baseline covariate $X_{-1}$ and a treatment effect that does not vary with $X_{-1}$, we write $\delta$ and $\lambda$ for the treatment effect $\delta_T(\bar x)$ and persistence factor $\lambda_T(\bar x)$ of Section~\ref{strucmod}, dropping the covariate argument, which is constant across $X_{-1}$ throughout.
	
	Panel A uses the following data generating process (DGP) with independent period-specific unobservables:
	\begin{align*}
		&U,\, V_0,\, V_1,\, W_0,\, W_1,\, X_{-1} \;\sim\; N(0,1)\;\text{mutually independent},\\
		&D_0 = X_{-1} + W_0 + U,\quad Y_0 = D_0 + U + V_0,\\
		&D_1 = X_{-1} + W_1 + \alpha U,\quad Y_1 = D_1 + \gamma U + X_{-1} + V_1.
	\end{align*}
	We report results for four linear estimators: IV-$\nabla D$ (OLS of $Y_1$ on $\nabla D_1$ and $X_{-1}$), ($\nabla$Y) (OLS of $\nabla Y_1$ on $D_1$, $D_0$, and $X_{-1}$), (D) (OLS of $Y_1$ on $D_1$, $D_0$, $X_{-1}$, and $Y_0$), and TWFE (OLS of $\nabla Y_1$ on $\nabla D_1$ and $X_{-1}$).
	
	The moment condition (IV) requires the treatment change to be uncorrelated with the structural residual: $\nabla D_1 = W_1 - W_0 + (\alpha-1)U$ contains $U$ whenever $\alpha \neq 1$, but this only induces a nonzero covariance with the residual if $U$ also affects $Y_1$, i.e.\ if $\gamma \neq 0$. (IV) therefore holds iff $\alpha=1$ or $\gamma=0$. ($\nabla$Y) requires the treatment level to be uncorrelated with the residualized outcome change: $\nabla Y_1 - \delta(D_1-D_0) = (\gamma-1)U + X_{-1} + V_1 - V_0$ correlates with $D_1$ whenever $\gamma \neq 1$ and $\alpha \neq 0$, so ($\nabla$Y) holds iff $\gamma=1$ or $\alpha=0$. (D) requires no residual correlation after conditioning on $Y_0$, $D_0$, and $X_{-1}$: when $\gamma=0$, $U$ does not enter $Y_1$ so there is no correlation regardless of $\alpha$; when $\alpha=0$, $U$ does not enter $D_1$ so $D_1$ is exogenous regardless of $\gamma$. (D) therefore holds iff $\alpha=0$ or $\gamma=0$.
	
	TWFE's DR property rests on a different, and in general non-equivalent, pair of conditions. Route (i) of Proposition~\ref{prop:doublerobust} requires $U$ to enter the treatment equation with the \emph{same} coefficient in both periods --- here, period $0$'s coefficient is fixed at $1$, so route (i) requires $\alpha=1$, regardless of $\gamma$. Route (ii) requires $U$ to enter the outcome equation with the same coefficient in both periods --- here, requiring $\gamma=1$, regardless of $\alpha$. Note that $\delta=1$ in both periods throughout, so the homogeneity condition of Proposition~\ref{prop:doublerobust} is otherwise satisfied whenever one of these two conditions holds. The condition $\gamma=0$, which makes (IV) and (D) hold, does not by itself satisfy either route: it removes $U$ from the period-$1$ outcome residual entirely, rather than matching its coefficient to period $0$'s, where $U$ still enters with coefficient $1$.
	
	We consider six scenarios: A1 ($\alpha=1, \gamma=2$), in which (IV) holds ($\alpha=1$) but ($\nabla$Y) fails ($\gamma \neq 1$, $\alpha \neq 0$) and (D) fails ($\alpha \neq 0$, $\gamma \neq 0$); route (i) for TWFE's DR applies since $\alpha=1$. A2 ($\alpha=1, \gamma=0$), in which (IV) and (D) hold since $\alpha=1$ and $\gamma=0$ respectively, but ($\nabla$Y) fails since $\gamma \neq 1$ and $\alpha \neq 0$; route (i) for TWFE's DR again applies since $\alpha=1$. A3 ($\alpha=1, \gamma=1$), in which (IV) and ($\nabla$Y) hold simultaneously, and both DR routes apply. A4 ($\alpha=2, \gamma=1$), in which ($\nabla$Y) holds ($\gamma=1$) but (IV) fails ($\alpha \neq 1$, $\gamma \neq 0$) and (D) fails; route (i) for TWFE's DR fails since $\alpha\neq1$, but route (ii) applies since $\gamma=1$. A5 ($\alpha=2, \gamma=0$), in which (IV) and (D) hold since $\gamma=0$ eliminates $U$ from $Y_1$, but ($\nabla$Y) fails since $\gamma \neq 1$ and $\alpha \neq 0$; \emph{neither} DR route applies, since $\alpha\neq1$ and $\gamma\neq1$, even though (IV) holds. A6 ($\alpha=2, \gamma=2$), in which all three moment conditions are violated and neither DR route applies either.
	
	\begin{table}[ht]
		\centering
		\caption{Panel A simulation results (independent $W_0, W_1$).\label{tab:simA}}
		\small
		\setlength{\tabcolsep}{4pt}
		\begin{tabular}{rlcccccccc}
			\toprule
			& & \multicolumn{2}{c}{IV-$\nabla D$} & \multicolumn{2}{c}{($\nabla$Y)}
			& \multicolumn{2}{c}{(D)} & \multicolumn{2}{c}{TWFE} \\
			\cmidrule(lr){3-4}\cmidrule(lr){5-6}\cmidrule(lr){7-8}\cmidrule(lr){9-10}
			& Scenario & Bias & RMSE & Bias & RMSE & Bias & RMSE & Bias & RMSE \\
			\midrule
			\multicolumn{10}{l}{\textit{$n = 500$}} \\
			\midrule
			A1 & (IV) only, route (i) $(\alpha=1,\gamma=2)$
			& $0.01$ & $0.21$ & $0.33$ & $0.34$ & $0.50$ & $0.50$ & $0.00$ & $0.06$ \\
			A2 & (IV) \& (D), route (i) $(\alpha=1,\gamma=0)$
			& $-0.00$ & $0.10$ & $-0.34$ & $0.34$ & $-0.00$ & $0.04$ & $-0.00$ & $0.06$ \\
			A3 & (IV) \& ($\nabla$Y), both routes $(\alpha=1,\gamma=1)$
			& $-0.00$ & $0.14$ & $-0.00$ & $0.05$ & $0.25$ & $0.25$ & $0.00$ & $0.04$ \\
			A4 & ($\nabla$Y) only, route (ii) $(\alpha=2,\gamma=1)$
			& $0.33$ & $0.34$ & $0.00$ & $0.04$ & $0.29$ & $0.29$ & $0.00$ & $0.04$ \\
			A5 & (IV) \& (D), neither route $(\alpha=2,\gamma=0)$
			& $0.00$ & $0.04$ & $-0.34$ & $0.34$ & $0.00$ & $0.03$ & $-0.34$ & $0.34$ \\
			A6 & All invalid, neither route $(\alpha=2,\gamma=2)$
			& $0.67$ & $0.67$ & $0.33$ & $0.34$ & $0.57$ & $0.57$ & $0.33$ & $0.34$ \\
			\midrule
			\multicolumn{10}{l}{\textit{$n = 2{,}000$}} \\
			\midrule
			A1 & (IV) only, route (i) $(\alpha=1,\gamma=2)$
			& $-0.00$ & $0.10$ & $0.33$ & $0.33$ & $0.50$ & $0.50$ & $-0.00$ & $0.03$ \\
			A2 & (IV) \& (D), route (i) $(\alpha=1,\gamma=0)$
			& $0.00$ & $0.05$ & $-0.33$ & $0.34$ & $-0.00$ & $0.02$ & $0.00$ & $0.03$ \\
			A3 & (IV) \& ($\nabla$Y), both routes $(\alpha=1,\gamma=1)$
			& $0.00$ & $0.07$ & $-0.00$ & $0.03$ & $0.25$ & $0.25$ & $-0.00$ & $0.02$ \\
			A4 & ($\nabla$Y) only, route (ii) $(\alpha=2,\gamma=1)$
			& $0.33$ & $0.33$ & $0.00$ & $0.02$ & $0.29$ & $0.29$ & $0.00$ & $0.02$ \\
			A5 & (IV) \& (D), neither route $(\alpha=2,\gamma=0)$
			& $-0.00$ & $0.02$ & $-0.33$ & $0.33$ & $-0.00$ & $0.01$ & $-0.33$ & $0.33$ \\
			A6 & All invalid, neither route $(\alpha=2,\gamma=2)$
			& $0.66$ & $0.67$ & $0.33$ & $0.33$ & $0.57$ & $0.57$ & $0.33$ & $0.33$ \\
			\bottomrule
		\end{tabular}
		\medskip\\
		\begin{minipage}{\textwidth}
			\small\textit{Notes:}
			DGP: $D_0=X_{-1}+W_0+U$, $D_1=X_{-1}+W_1+\alpha U$,
			$Y_1=D_1+\gamma U+X_{-1}+V_1$.
			(IV) valid iff $\alpha=1$ or $\gamma=0$; ($\nabla$Y) valid iff $\gamma=1$ or $\alpha=0$; (D) valid iff $\alpha=0$ or $\gamma=0$. TWFE is DR (Proposition~\ref{prop:doublerobust}) via route (i) iff $\alpha=1$ and via route (ii) iff $\gamma=1$. Neither route coincides exactly with (IV) or ($\nabla$Y). 
			Bias and root mean squared error (RMSE) are computed against a common target $\delta=1$ for all four estimators. For IV-$\nabla D$, the raw regression coefficient converges to $\delta\cdot\lambda_1$ rather than $\delta$. We therefore rescale it by $1/\hat\lambda_1$, estimated separately within each scenario, before computing bias and RMSE, as for the OLS estimate in the empirical application below. This isolates the bias attributable to a violation of condition (IV) from the mechanical effect of $\lambda_1$ varying with $\alpha$ across scenarios.
		\end{minipage}
	\end{table}
	
	The probability limits against which bias is measured follow directly from equation~\eqref{eq:ivD_estimand}: under (IV), the IV-$\nabla D$ regression coefficient converges to $\delta\cdot\lambda_1$, not to $\delta$ itself, while ($\nabla$Y), (D), and TWFE converge to $\delta$ directly when their respective moment conditions hold. We therefore rescale the IV-$\nabla D$ coefficient by $1/\hat\lambda_1$, as discussed in Section~\ref{sec:overid} and applied to the OLS estimate in the empirical application below, and report bias for the rescaled estimate against the common target $\delta=1$, on the same footing as the other three estimators. The persistence factor $\lambda_1$ depends on $\alpha$ in this DGP: when $\alpha=1$, $\nabla D_1 = W_1 - W_0$ with $W_0$ and $W_1$ independent and of equal variance, so $\lambda_1 = \mathrm{Var}(W_1)/(\mathrm{Var}(W_1)+\mathrm{Var}(W_0)) = 1/2$; when $\alpha=2$, $\lambda_1=1$ instead. Rescaling by the scenario-specific $\hat\lambda_1$ in each case isolates the bias attributable to a genuine violation of condition (IV) from the mechanical effect of $\lambda_1$ differing from one, which would otherwise contaminate the comparison across scenarios with different values of $\alpha$. TWFE requires no such rescaling because differencing both outcome and treatment yields $\nabla Y_1 = \delta\,\nabla D_1 + \nabla \eta_1$ whenever $\delta$ is homogeneous across periods, as it is throughout this DGP. The TWFE coefficient therefore converges to $\delta=1$ whenever $\nabla\eta_1$ is uncorrelated with $\nabla D_1$, regardless of $\lambda_1$. TWFE bias, as well as ($\nabla$Y) and (D) biases, are reported against $\delta=1$.
	
	Table~\ref{tab:simA} reports bias and root mean squared error (RMSE) based on 1{,}000 simulations and two sample sizes, $n=500$ and $n=2{,}000$. For all consistent estimators, the RMSE approximately halves when the sample size quadruples, consistent with the expected $\sqrt{n}$ convergence rate, while the bias remains negligible across both sample sizes. Scenario A1 shows that (IV) does not imply ($\nabla$Y) or (D): the rescaled IV-$\nabla D$ estimate has a bias of $\approx 0$ against $\delta=1$, while ($\nabla$Y) has a persistent bias of $0.33$ and (D) has a bias of $0.50$. TWFE also has a bias of $\approx 0$ against $\delta=1$: this reflects route (i) for TWFE's DR, which applies because $\alpha=1$, rather than the moment condition (IV) as such.
	
	Scenarios A4 and A5 both have $\alpha=2$, so $\lambda_1=1$ in both, and the rescaled IV-$\nabla D$ estimate isolates the bias attributable to condition (IV) alone. At A5, where $\gamma=0$ makes (IV) hold, the rescaled bias is $\approx 0$, confirming that rescaling by the correct, scenario-specific $\hat\lambda_1$ removes any contamination from $\lambda_1$ differing from one. At A4, where $\gamma=1\neq0$ makes (IV) fail, the rescaled bias is $0.33$, reflecting the genuine violation. The two scenarios diverge sharply for TWFE, however: it is unbiased at A4 but has a bias of $-0.34$ at A5, even though (IV) holds at A5 and fails at A4. The reason is that, in this DGP, $U$ enters the period-$0$ outcome with coefficient $1$, so route (ii) for TWFE's DR requires $\gamma=1$ specifically: at A4, $\gamma=1$ matches this coefficient exactly, and TWFE is robust through route (ii); at A5, $\gamma=0$ instead removes $U$ from the period-$1$ residual without matching it to period $0$, so neither DR route applies, and TWFE is biased even though the single-period condition (IV) holds. Scenario A5 thus shows that satisfying (IV) alone --- a restriction on the period-$1$ residual only --- does not confer DR on TWFE when $U$'s coefficient in the treatment equation also changes across periods. Scenarios A2 and A3 illustrate further non-nesting directions of Proposition~\ref{prop:nonnest_moments}. The remaining direction, (D) holding while (IV) fails, requires $\alpha=0$ in this DGP and is covered by Example~\ref{ex:iv_vs_d}. Scenario A6 confirms that all four estimators are biased when all moment conditions are violated simultaneously.
	
	Panel B uses a DGP in which the time-varying treatment unobservable $W_T$ follows a random walk and a dynamic treatment effect is present:
	\begin{align*}
		&U,\, V_0,\, V_1,\, W_0,\, \varepsilon_1,\, X_{-1}  \;\sim\; N(0,1)\;\text{mutually independent},\\
		&D_0 = X_{-1} + U + W_0,\quad W_1 = W_0 + \varepsilon_1,\quad D_1 = X_{-1} + U + W_1,\\
		&Y_0 = D_0 + U + V_0,\quad Y_1 = D_1 + \phi D_0 + \gamma U + X_{-1} + V_1,
	\end{align*}
	where $\phi \geq 0$ controls the dynamic effect of $D_0$ on $Y_1$. This DGP satisfies the treatment equation of Model A, with $\nabla D_1 = \varepsilon_1$ always free of $U$ by the random-walk structure, so by Proposition~\ref{modelA_all} Assumptions~\ref{CIAdD}, \ref{CIAD}, and~\ref{CIAdY} all hold simultaneously, regardless of $\gamma$ or $\phi$, and TWFE inherits consistency. Since $\lambda_1=1$ here, IV-$\nabla D$ is not attenuated and all four estimators are reported against $\delta=1$. We consider two scenarios: B1 ($\gamma=1, \phi=0$), a baseline without dynamic effects and moderate confounding, in which all strategies are straightforwardly valid; and B2 ($\gamma=2, \phi=0.5$), which combines strong fixed-effect confounding with a dynamic treatment effect that would invalidate strategies under the Panel A DGP, but does not do so here because $\nabla D_1 = \varepsilon_1$ is a genuinely fresh innovation, independent of $U$ and of the entire treatment history.
	
	\begin{table}[ht]
		\centering
		\caption{Panel B simulation results (random walk of $W_T$, Model A).\label{tab:simB}}
		\small
		\setlength{\tabcolsep}{4pt}
		\begin{tabular}{rlcccccccc}
			\toprule
			& & \multicolumn{2}{c}{IV-$\nabla D$} & \multicolumn{2}{c}{($\nabla$Y)}
			& \multicolumn{2}{c}{(D)} & \multicolumn{2}{c}{TWFE} \\
			\cmidrule(lr){3-4}\cmidrule(lr){5-6}\cmidrule(lr){7-8}\cmidrule(lr){9-10}
			& Scenario & Bias & RMSE & Bias & RMSE & Bias & RMSE & Bias & RMSE \\
			\midrule
			\multicolumn{10}{l}{\textit{$n = 500$}} \\
			\midrule
			B1 & No dynamic effect $(\gamma=1,\phi=0)$
			& $-0.00$ & $0.11$ & $-0.00$ & $0.07$ & $-0.00$ & $0.05$ & $-0.00$ & $0.07$ \\
			B2 & Dynamic effect $(\gamma=2,\phi=0.5)$
			& $-0.00$ & $0.17$ & $-0.00$ & $0.07$ & $-0.00$ & $0.07$ & $-0.00$ & $0.09$ \\
			\midrule
			\multicolumn{10}{l}{\textit{$n = 2{,}000$}} \\
			\midrule
			B1 & No dynamic effect $(\gamma=1,\phi=0)$
			& $-0.00$ & $0.05$ & $-0.00$ & $0.03$ & $-0.00$ & $0.03$ & $-0.00$ & $0.03$ \\
			B2 & Dynamic effect $(\gamma=2,\phi=0.5)$
			& $-0.00$ & $0.09$ & $-0.00$ & $0.04$ & $-0.00$ & $0.03$ & $-0.00$ & $0.05$ \\
			\bottomrule
		\end{tabular}
		\medskip\\
		\begin{minipage}{\textwidth}
			\small\textit{Notes:}
			DGP: $D_0=X_{-1}+U+W_0$, $W_1=W_0+\varepsilon_1$, $D_1=X_{-1}+U+W_1$,
			$Y_1=D_1+\phi D_0+\gamma U+X_{-1}+V_1$. Treatment equation satisfies Model A; Assumptions~\ref{CIAdD}, \ref{CIAD}, and~\ref{CIAdY} all hold simultaneously regardless of $\gamma$ or $\phi$ (Proposition~\ref{modelA_all}), and TWFE inherits consistency. Bias and root mean squared error (RMSE) compared against each estimator's probability limit, $\delta = 1$.
		\end{minipage}
	\end{table}
	
	Table~\ref{tab:simB} confirms Proposition~\ref{modelA_all}: all four estimators are unbiased in both scenarios, with RMSE again approximately halving when going from $n=500$ to $n=2{,}000$, consistent with $\sqrt{n}$ convergence. Scenario B2 is the key illustration: despite strong fixed-effect confounding and a dynamic treatment effect, all four methods remain consistent under the random-walk structure of Model A. The larger RMSE of IV-$\nabla D$ and TWFE relative to ($\nabla$Y) and (D) in B2 reflects that ($\nabla$Y) and (D) include $D_0$ as a regressor, directly controlling for the dynamic effect $\phi D_0$ on $Y_1$ and thereby reducing residual variance, whereas IV-$\nabla D$ and TWFE do not include $D_0$ separately so this variation remains in the residual.
	
	The contrast between Panels A and B illustrates the central message of Section~\ref{strucmod}. In Panel A, the treatment process is no random walk for any value of $\alpha$, and the four estimators rely on different moment conditions. Each estimator can therefore be consistent while the others are biased. Under the random-walk structure of Model A, illustrated in Panel B, all four moment conditions hold simultaneously, and all four estimators are consistent for $\delta$.
	\section{Empirical Application: Cigarette Demand}\label{app}
	
	To illustrate the identification strategies and the overidentification tests developed in Sections~\ref{setup}--\ref{sec:overid}, we estimate the price elasticity of cigarette demand using the panel data previously analyzed by \citet{baltagi1994} covering 46 US states over the period 1963--1992. The data set covers $1{,}380$ state-year observations (or $1{,}334$ after first-differencing) and is available in the \texttt{plm} package by \citet{croissant2008plm} for the statistical software \texttt{R}. The outcome $Y_{i,t}$ for individual $i$ in period $t$ is log per-capita cigarette sales (packs per year), the treatment $D_{i,t}$ is log real cigarette price (nominal price deflated by the CPI), and the covariates $X_{i,t}$ consist of log real per-capita disposable income and log population aged 16 and over. Cigarette demand is an interesting testing ground, as real prices vary substantially across states and over time owing to differences in state excise taxes, and the literature offers a well-established benchmark price elasticity of approximately $-0.3$ to $-0.5$, as discussed in \citet{chaloupkaWarner2000}.
	
	We implement four empirical strategies: the instrumental-variables estimator motivated by Proposition~\ref{prop:modelBiv}, which we label IV-$\nabla D$; the DiD-type estimator based on Assumption~\ref{CIAdY}, which we label ($\nabla$Y); the selection-on-observables estimator based on Assumption~\ref{CIAD}, which we label (D); and TWFE. In line with our previous discussion, for any strategy in which $D_{i,t}$ enters in levels, $D_{i,t-1}$ is added as an additional control variable, and whenever $Y_{i,t}$ enters in levels, $Y_{i,t-1}$ is added as an additional control variable. The continuous controls $X_{i,t}$ enter at both the current level $X_{i,t}$ and the lagged level $X_{i,t-1}$ as separate regressors. Including $X_{i,t}$ assumes that the cigarette price in year $t$ does not affect disposable income or population in the same year. Year fixed effects are included in all four specifications. State fixed effects are never included explicitly, but are eliminated by outcome differencing when relying on ($\nabla$Y) and TWFE, while $Y_{i,t-1}$ is controlled for instead when relying on IV-$\nabla D$ and (D). Assuming linear outcome models, the four regression specifications are:
	\begin{itemize}
		\item ($\nabla$Y):\; $\nabla Y_{i,t} \sim D_{i,t} + D_{i,t-1} + X_{i,t} + X_{i,t-1} + \text{year FE}$
		\item IV-$\nabla D$:\; $Y_{i,t} \sim \nabla D_{i,t} + Y_{i,t-1} + X_{i,t} + X_{i,t-1} + \text{year FE}$
		\item (D):\; $Y_{i,t} \sim D_{i,t} + Y_{i,t-1} + D_{i,t-1} + X_{i,t} + X_{i,t-1} + \text{year FE}$
		\item TWFE:\; $\nabla Y_{i,t} \sim \nabla D_{i,t} + X_{i,t} + X_{i,t-1} + \text{year FE}$
	\end{itemize}
	
	As established in Section~\ref{sec:overid}, ($\nabla$Y), (D), and TWFE all estimate the homogeneous price elasticity $\delta$ directly under their respective moment conditions, while the IV-$\nabla D$ coefficient converges to $\delta\cdot\lambda$ and must be rescaled by the inverse of an estimate of the persistence factor $\lambda$, denoted by $\hat\lambda$, before overidentification tests can be run. Here $\delta$ and $\lambda$ denote the pooled, covariate-independent versions of $\delta_T(\bar x)$ and $\lambda_T(\bar x)$ from Section~\ref{strucmod}: the OLS specifications below impose a single price elasticity and a single persistence factor common to all states and years, rather than allowing either to vary with $\bar x$. We estimate $\hat\lambda$ as the pooled OLS coefficient from a regression of $D_{i,t}$ on $\nabla D_{i,t}$ after partialling out $Y_{i,t-1}$, $X_{i,t-1}$, and year fixed effects, matching the definition in equation~\eqref{eq:lambda_def} averaged over the joint distribution of states and years:
	\begin{align}
		\hat\lambda = \frac{\widehat{\mathrm{Cov}}(\widetilde{\nabla D}_{i,t},\, \widetilde{D}_{i,t})}{\widehat{\mathrm{Var}}(\widetilde{\nabla D}_{i,t})},
	\end{align}
	where tildes denote residuals after partialling out $Y_{i,t-1}$, $X_{i,t-1}$, and year fixed effects. We estimate $\hat\lambda$ conditional on pre-treatment controls only, as in definition~\eqref{eq:lambda_def}. Including the contemporaneous controls $X_{i,t}$ changes $\hat\lambda$ only marginally, from $0.553$ to $0.555$. The standard error of $\hat\lambda$ is computed by block-bootstrapping states with 500 bootstrap samples, and the resulting uncertainty is propagated into the standard error of the rescaled IV-$\nabla D$ estimate via the delta method.
	
	In addition to OLS, we estimate treatment effects using the causal forest of \citet{WagerAthey2018} and \citet{AtheyTibshiraniWager2019}, which allows for heterogeneous treatment effects and nonparametric estimation. For ($\nabla$Y) and (D), the treatment enters in levels, and the causal forest estimates the average over covariate values of the conditional slope of the outcome in the price, if Assumptions~\ref{CIAdY} and~\ref{CIAD} are satisfied, respectively. This slope equals the conditional average marginal effect $\delta_t(\bar x)$ in equation~\eqref{eq:ame} if the price response is linear within covariate cells or the price is normally distributed given the covariates. Otherwise, it is a weighted average of marginal effects with nonnegative weights. For IV-$\nabla D$, Proposition~\ref{prop:modelBiv} guarantees that the ratio in equation~\eqref{eq:ivratio} recovers the price elasticity exactly when it is homogeneous, and Proposition~\ref{prop:doublerobust} establishes the analogous result for TWFE when the price elasticity takes the same value in both periods being differenced. Both results are specific to a price elasticity that does not vary with the price level itself. If, however, the price elasticity is heterogeneous in the price level, the IV-$\nabla D$ ratio identifies a weighted average of the marginal price response rather than $\delta_t(\bar x)$, as discussed in Section~\ref{strucmod}. Its weights are nonnegative if the distribution of price levels net of the state fixed effect, conditional on the covariates, is monotonic in the price change. Furthermore, TWFE's DR property no longer applies under effect heterogeneity, which needs to be borne in mind when interpreting the results further below. 	
	
	For IV-$\nabla D$, we train an instrumental forest with $Y_{i,t}$ as the outcome, $D_{i,t}$ as the treatment, and $\nabla D_{i,t}$ as the instrument, conditional on the same covariates as the corresponding OLS specification, including year dummies. It estimates the covariate-specific ratio $\delta(\bar x)=\mathrm{Cov}(\nabla D_{i,t},Y_{i,t}\mid \bar x)/\mathrm{Cov}(\nabla D_{i,t},D_{i,t}\mid \bar x)$, so that the persistence factor may vary with the covariates. We average $\hat\delta(\bar x)$ using doubly robust scores, which require a second forest for the first stage $\lambda(\bar x)$, with $D_{i,t}$ as the outcome and $\nabla D_{i,t}$ as the treatment. For TWFE, we train a single causal forest with $\nabla Y_{i,t}$ as the outcome and $\nabla D_{i,t}$ as the treatment conditional on the covariates. We use the \texttt{grf} package of \citet{TibshiraniAtheyWager2020grf} with 2{,}000 honest trees, tuned hyperparameters, and clusters at the state level throughout.
	
	\begin{table}[!htbp]
		\centering
		\caption{Cigarette Demand: Price Elasticity under Four
			Identification Strategies. Baltagi \& Levin (1986) panel, 46 US states $\times$ 29
			years (1964--1992), $n=1{,}334$.\label{tab:cigar}}
		\small
		\begin{tabular}{lcccc}
			\toprule
			& \multicolumn{2}{c}{OLS} & \multicolumn{2}{c}{Causal Forest} \\
			\cmidrule(lr){2-3}\cmidrule(lr){4-5}
			Strategy & Estimate & Standard error & Estimate & Standard error \\
			\midrule
			($\nabla$Y) &
			$-0.395^{***}$ & $(0.038)$ &
			$-0.300^{***}$ & $(0.038)$ \\[2pt]
			(D) &
			$-0.402^{***}$ & $(0.039)$ &
			$-0.458^{***}$ & $(0.072)$ \\[2pt]
			TWFE &
			$-0.380^{***}$ & $(0.039)$ &
			$-0.366^{***}$ & $(0.036)$ \\[2pt]
			IV-$\nabla D$ &
			$-0.684^{***}$ & $(0.078)$ &
			$-0.691^{***}$ & $(0.092)$ \\
			\midrule
			\multicolumn{5}{l}{\textit{Hausman overidentification tests (OLS $p$-value\;$|$\;Causal Forest $p$-value)}} \\[3pt]
			\multicolumn{1}{l}{($\nabla$Y) vs IV-$\nabla D$} &
			\multicolumn{2}{c}{$p<0.001^{***}$} &
			\multicolumn{2}{c}{$p<0.001^{***}$} \\
			\multicolumn{1}{l}{($\nabla$Y) vs (D)} &
			\multicolumn{2}{c}{$p=0.074$} &
			\multicolumn{2}{c}{$p=0.063$} \\
			\multicolumn{1}{l}{($\nabla$Y) vs TWFE} &
			\multicolumn{2}{c}{$p=0.002^{**}$} &
			\multicolumn{2}{c}{$p=0.023^{*}$} \\
			\multicolumn{1}{l}{IV-$\nabla D$ vs (D)} &
			\multicolumn{2}{c}{$p<0.001^{***}$} &
			\multicolumn{2}{c}{$p=0.015^{*}$} \\
			\multicolumn{1}{l}{IV-$\nabla D$ vs TWFE} &
			\multicolumn{2}{c}{$p<0.001^{***}$} &
			\multicolumn{2}{c}{$p<0.001^{***}$} \\
			\multicolumn{1}{l}{(D) vs TWFE} &
			\multicolumn{2}{c}{$p<0.001^{***}$} &
			\multicolumn{2}{c}{$p=0.161$} \\
			\bottomrule
		\end{tabular}
		\medskip\\
		\begin{minipage}{0.95\textwidth}
			\small\textit{Notes:} $Y_{i,t}$\,=\,log per-capita cigarette sales; $D_{i,t}$\,=\,log real cigarette price; $X_{i,t}$\,=\,log real disposable income, log population aged 16+. Controls enter at levels in both periods $t$ and $t-1$. All specifications include year fixed effects; state fixed effects are absorbed by outcome differencing for ($\nabla$Y) and TWFE. Standard errors clustered by state in parentheses; all causal forests are trained with state clusters. The OLS and causal forest coefficients for ($\nabla$Y), (D), and TWFE estimate $\delta$ directly. For IV-$\nabla D$, the raw OLS coefficient ($-0.378$) corresponds to $\hat\delta \cdot \hat\lambda$ and is rescaled by $1/\hat\lambda = 1.81$, with the standard error obtained by the delta method. The causal forest estimate for IV-$\nabla D$ is the average of the covariate-specific ratios $\hat\delta(\bar x)$ of an instrumental forest with $D_{i,t}$ as treatment and $\nabla D_{i,t}$ as instrument, computed from doubly robust scores. The causal forest row for IV-$\nabla D$ identifies a weighted average of the marginal price response that generally differs from the averages estimated for ($\nabla$Y) and (D), since the random-walk property needed for the IV-$\nabla D$ estimand to coincide with the others is rejected by the data; with a heterogeneous price response, the causal forest estimate for TWFE averages over covariate values the weighted average of marginal price responses discussed in Section~\ref{strucmod}. Hausman $p$-values from a two-sided $z$-test of the difference between two estimates. Its standard error accounts for the covariance between the estimates and is obtained from a state-block bootstrap with 500 replications that re-estimates all four strategies (OLS) and from state-clustered doubly robust scores (causal forest). $^{*}p<0.05$, $^{**}p<0.01$, $^{***}p<0.001$.
		\end{minipage}
	\end{table}
	
	Table~\ref{tab:cigar} reports the estimation results. The linear estimate of the persistence factor is $\hat\lambda = 0.553$ with a standard error (SE) of $0.027$ based on the block bootstrap, strongly rejecting the random-walk hypothesis $\lambda=1$ ($p<0.001$). This implies that Model A does not describe the price process, motivating the instrument-based approach underlying  IV-$\nabla D$, for which the rescaling factor is $1/\hat\lambda = 1.81$ for the linear case. 
	The OLS estimates based on ($\nabla$Y) ($\hat\delta=-0.395$, SE$\,=0.038$), (D) ($\hat\delta=-0.402$, SE$\,=0.039$), and TWFE ($\hat\delta=-0.380$, SE$\,=0.039$) are tightly grouped, while the rescaled IV-$\nabla D$ estimate ($\hat\delta=-0.684$, SE$\,=0.078$) exceeds this cluster in absolute value by approximately $0.29$. Running Hausman overidentification tests reflects this finding: all three comparisons involving IV-$\nabla D$ reject ($p<0.001$ in each case). Among the remaining comparisons, ($\nabla$Y) vs.\ (D) does not reject at the 5\% level ($p=0.074$), while ($\nabla$Y) vs.\ TWFE ($p=0.002$) and (D) vs.\ TWFE ($p<0.001$) reject. The latter differences are small, at $0.014$ and $0.022$, and are statistically significant because the estimates are highly correlated, as they are computed from the same sample. The level-based strategies and TWFE therefore agree closely in magnitude, but the small differences between them are statistically detectable. 
	
	The causal forest estimates are less precise and more dispersed. The estimates based on ($\nabla$Y) ($\hat\delta=-0.300$, SE$\,=0.038$) and (D) ($\hat\delta=-0.458$, SE$\,=0.072$) differ more than under OLS, but the corresponding Hausman test does not reject at the 5\% level ($p=0.063$). The causal forest estimate for TWFE ($\hat\delta=-0.366$, SE$\,=0.036$) lies between them. It differs significantly from ($\nabla$Y) ($p=0.023$) but not from (D) ($p=0.161$). The instrumental-forest estimate for IV-$\nabla D$ ($\hat\delta=-0.691$, SE$\,=0.092$) is very close to its OLS counterpart and differs significantly from all three other strategies ($p<0.001$ against ($\nabla$Y) and TWFE, $p=0.015$ against (D)). Its first-stage estimates $\hat\lambda(\bar x)$ range from $0.73$ to $0.77$, above the linear estimate $\hat\lambda=0.553$, which controls for the covariates linearly. They remain clearly below one, in line with the rejection of the random walk.
	
	Taken together, these results indicate that the price-change-based instrumental-variables strategy tends to disagree with the other three strategies in this application: across both OLS and the causal forest, Hausman tests consistently find that the IV-$\nabla D$ estimate is distinct from those based on ($\nabla$Y), (D), and TWFE. These three agree closely with each other in magnitude under OLS, and less closely under the causal forest. In a linear world in which the price elasticity  does not vary with the price level, the gap between the rescaled linear IV-$\nabla D$ estimate and the other three strategies based on OLS indicates a violation of the IV-exogeneity condition~\eqref{eq:ivexog}. The comparison with TWFE is especially informative in light of Proposition~\ref{prop:doublerobust}, since the linear TWFE estimate remains valid whenever the state fixed effect cancels under differencing in either the price equation or the demand equation, given a price elasticity that is constant over time. A plausible source of such a violation is that states experiencing large year-to-year price increases differ systematically in ways that past outcomes and covariates do not fully capture, plausibly because tax-driven price changes reflect state-level political and fiscal conditions that can also be associated with cigarette consumption.  
	
	If the price elasticity varies with the covariates, the causal forest is arguably more appropriate as an estimation method due to its flexibility in the covariates. Within covariate cells, however, it remains linear in the price, so that a nonlinear price response enters its estimates through the weighted averages discussed above. When considering the causal forest results, the gap between IV-$\nabla D$ and ($\nabla$Y) and (D) remains substantial and statistically significant. In a nonlinear model, the difference between the IV-$\nabla D$ ratio and the level-based approaches can then either be driven by heterogeneity in the price elasticity across price levels, a violation of some identifying assumptions such as condition~\eqref{eq:ivexog}, or both. Without imposing linearity, we are therefore unable to rule out either explanation on the basis of these estimates.
	
	\section{Conclusion}\label{conc}
	
	This paper has clarified when exploiting treatment changes rather than levels identifies causal effects. The central difficulty is that a unit's treatment change is not assigned independently of its history: it builds on whatever treatment level the unit already had, and that level reflects the same forces that may also drive the outcome. Removing a fixed, time-constant confounder is not enough to rule this out, since the outcome can still depend on the unit's evolving treatment history through channels a simple differencing argument does not address.
	
	We show that a treatment process whose unexplained, time-varying component follows a random walk resolves this fully: because each period brings a genuinely new shock, unrelated to everything before it, the treatment change carries no information about the unit's history, and heterogeneous, dynamic treatment effects are identified without further restrictions on the outcome. A second model, which does not require the random walk but otherwise leaves the treatment process unrestricted, does not achieve this on its own, as we show by explicit example. Conditioning on the lagged treatment level does not repair it either. We show, however, that if dynamic treatment effects are ruled out and the treatment effect is taken to be constant given covariates, this second model still makes the treatment change a valid instrument for the treatment level, identifying that constant effect via the standard instrumental-variables ratio. The two models are not nested --- each relaxes one restriction while imposing another.
	
	The random walk permits identifying the conditional average marginal effect, for any pattern of treatment effect heterogeneity. The second model is built around the opposite case, a treatment effect taken to be constant given covariates, and its instrumental-variables route is designed for that case. If the treatment effect is in fact heterogeneous across treatment levels, this same route still identifies a weighted average of the marginal response, as a conventional instrument would under treatment effect heterogeneity, but generally not the same conditional average marginal effect that a selection-on-observables or DiD strategy targets directly. Its weights are nonnegative if, given covariates, the distribution of the treatment level net of the time-constant confounder is monotonic in the treatment change. TWFE does not escape this limitation either: its clean recovery of the conditional average marginal effect, established in this paper for the linear case, also relies on conditional effect homogeneity. Once treatment effects vary with the treatment level, TWFE identifies a weighted average of marginal effects whose weights are nonnegative under an additional condition on the treatment process.
	
	We further show that, under the random walk, the treatment-change strategy coincides with two conventional strategies based on treatment levels --- selection-on-observables and difference-in-differences --- motivating Hausman-type overidentification tests, while at the level of the linear regression coefficients used in practice the three strategies remain genuinely non-nested. We establish a double-robustness (DR) property of the two-way fixed effects estimator for a treatment effect that is constant across periods: it remains consistent if the time-constant confounder enters either the treatment process or the outcome process in a way that cancels exactly under differencing. The first of these two conditions strengthens the instrumental-variables condition of the second model, and the second strengthens the parallel-trends condition underlying difference-in-differences. Both strengthenings arise because TWFE differences the outcome and therefore also depends on the outcome in the earlier period. The treatment change must then also be uncorrelated with the unexplained part of the earlier outcome, or parallel trends must also hold across units with different earlier treatment levels. It suffices that one of the two conditions holds, but this DR property does not extend to more general models once treatment effects vary across periods or across treatment levels.
	
	These results matter for empirical practice. Researchers using treatment changes for identification should be explicit about which structural route they rely on, since the two are not interchangeable. A natural diagnostic is to test whether the treatment process lacks the persistence implied by a random walk. Rejecting this, as we find for cigarette prices in our empirical application, means the treatment-change estimator should be read through the instrumental-variables lens rather than as recovering the same effect as the other strategies. The overidentification tests developed here offer a further diagnostic, best combined with falsification tests on pre-treatment periods, and the DR property of TWFE provides some protection against misspecification when strategies disagree, though it does not substitute for direct examination of the underlying structural assumptions.
	
	\pagebreak
	\begin{appendices}
		\section{Decomposition Result}\label{app:decomposition}
		
		This appendix states and proves the decomposition result referred to in Section~\ref{strucmod}, which makes precise the algebraic relationship between Model A's random-walk property and Model B's instrumental-variables condition in the partially linear model~\eqref{partiallylinear}.
		
		Since $D_T = D_{T-1}+\nabla D_T$ by definition, we have the algebraic identity, for any $\bar x$ in the support of $\bar X_{T-1}$,
		\begin{align}\label{eq:identity}
			\mathrm{Cov}(\nabla D_T,D_{T-1}\mid\bar x) = \mathrm{Cov}(\nabla D_T,D_T\mid\bar x) - \mathrm{Var}(\nabla D_T\mid\bar x) = \big(\lambda_T(\bar x)-1\big)\,\mathrm{Var}(\nabla D_T\mid\bar x),
		\end{align}
		which holds without any assumption, purely as a consequence of the definitions of $\nabla D_T$ and $\lambda_T(\bar x)$.
		
		\begin{proposition}[Decomposition of the linear CIA-$\nabla D$ requirement]\label{prop:decomposition}
			Under the outcome model~\eqref{partiallylinear}, the following identity holds for every $\bar x$ and every $a$, without further assumptions:
			\begin{align}\label{eq:masteridentity}
				\mathrm{Cov}\big(\nabla D_T,\,Y_T(D_{T-1}+a)\mid \bar X_{T-1}=\bar x\big) = \delta_T(\bar x)\big(\lambda_T(\bar x)-1\big)\,\mathrm{Var}(\nabla D_T\mid \bar x) + \mathrm{Cov}(\nabla D_T,\eta_T\mid \bar x).
			\end{align}
			Consequently, Assumption~\ref{CIAdD} --- which requires the left-hand side of~\eqref{eq:masteridentity} to vanish for all $\bar x$ and $a$ --- implies, at every $\bar x$ with $\delta_T(\bar x)\neq0$ and $\mathrm{Var}(\nabla D_T\mid\bar x)>0$,
			\begin{align}\label{eq:balance}
				\mathrm{Cov}(\nabla D_T,\eta_T\mid\bar x) = \delta_T(\bar x)\big(1-\lambda_T(\bar x)\big)\,\mathrm{Var}(\nabla D_T\mid\bar x).
			\end{align}
		\end{proposition}
		
		\begin{proof}
			Under~\eqref{partiallylinear}, $Y_T(D_{T-1}+a) = \delta_T(\bar x)(D_{T-1}+a) + \eta_T$. Conditional on $\bar X_{T-1}=\bar x$, the additive constant $\delta_T(\bar x)\,a$ does not affect the covariance, so
			\begin{align}
				\mathrm{Cov}\big(\nabla D_T,Y_T(D_{T-1}+a)\mid \bar x\big) = \delta_T(\bar x)\,\mathrm{Cov}(\nabla D_T,D_{T-1}\mid\bar x) + \mathrm{Cov}(\nabla D_T,\eta_T\mid\bar x).
			\end{align}
			Substituting~\eqref{eq:identity} for $\mathrm{Cov}(\nabla D_T,D_{T-1}\mid\bar x)$ gives~\eqref{eq:masteridentity}. Setting the left-hand side to zero, as required by Assumption~\ref{CIAdD}, and rearranging gives~\eqref{eq:balance}.
		\end{proof}
		
		Equation~\eqref{eq:balance} relates two quantities: the deviation of the persistence factor from one, $1-\lambda_T(\bar x)$, and the correlation of the treatment change with the structural residual, $\mathrm{Cov}(\nabla D_T,\eta_T\mid\bar x)$. It does not, by itself, force either quantity to vanish individually. In principle, a non-generic numerical coincidence between $\delta_T(\bar x)$, $\lambda_T(\bar x)$, and the distribution of the unobservables could satisfy~\eqref{eq:balance} with $\lambda_T(\bar x)\neq1$. As a numerical illustration, Example~\ref{ex:modelBfails} with $\rho=0$, with covariates $X_{-1}$ and $X_0$ added for concreteness ($D_0=X_{-1}+W_0+U$, $D_1=X_0+W_1+U$, $Y_1=D_1+\gamma U+X_0+V_1$, so that $\delta_1(\bar x)=1$ for all $\bar x$), satisfies $\mathrm{Cov}(\nabla D_1,\eta_1\mid \bar X_0)=0$ exactly, but has $\lambda_1(\bar X_0) = \mathrm{Var}(W_1)/(\mathrm{Var}(W_1)+\mathrm{Var}(W_0)) = 1/2$ when $W_0$ and $W_1$ have equal variance, so that equation~\eqref{eq:masteridentity} gives $\mathrm{Cov}(\nabla D_1,Y_1(D_0+a)\mid \bar X_0) = 1\cdot(1/2-1)\cdot 2\sigma_W^2 + 0 = -\sigma_W^2 \neq 0$: Assumption~\ref{CIAdD} fails whenever $\lambda_1\neq1$, even though the instrumental-variables condition holds.
		
		What equation~\eqref{eq:balance} does show is that this requirement is satisfied for every admissible choice of $\delta_T(\cdot)$ and every joint distribution of the unobservables --- rather than relying on such a coincidence --- only if \emph{both} quantities vanish simultaneously: $\lambda_T(\bar x)=1$, which the random walk of Model A implies, \emph{and} $\mathrm{Cov}(\nabla D_T,\eta_T\mid\bar x)=0$, the instrumental-variables condition of Model B. Model A guarantees both simultaneously, as a consequence of the full conditional independence in Proposition~\ref{prop:modelA}; Model B, on its own, guarantees only the second, and is therefore not, in general, sufficient for Assumption~\ref{CIAdD}, which is exactly Example~\ref{ex:modelBfails}'s point. The two models are accordingly not substitutable routes to Assumption~\ref{CIAdD}: Model A is sufficient for it, and Model B's contribution is one necessary, but on its own insufficient, ingredient of what the linear version of the assumption requires.
		
	\end{appendices}
	\pagebreak
	\bibliographystyle{apacite}
	\bibliography{refs}
	
\end{document}